\documentclass[11pt]{article}

\usepackage[margin=1in]{geometry}
\usepackage{amsmath,amsthm,amssymb,bbm,color,tikz,mathtools}
\usepackage[ruled,vlined]{algorithm2e}
\usetikzlibrary{positioning}
\usetikzlibrary{shapes}
\usetikzlibrary{decorations.pathreplacing}
\usepackage{enumerate, enumitem} 
\usepackage[pagebackref]{hyperref}
\hypersetup{colorlinks,linkcolor={blue},citecolor={blue},urlcolor={black}}
\allowdisplaybreaks

\newcommand{\cbra}[1]{\left\{ #1 \right\}}
\newcommand{\corr}[3]{\textnormal{corr}_{#3}(#1, #2)}
\newcommand{\discrepancy}[3]{\textnormal{\disc}_{#3}(#1, #2)}
\newcommand{\abs}[1]{\left\lvert #1 \right\rvert}
\newcommand{\norm}[1]{\left\lVert #1 \right\rVert}
\newcommand{\floor}[1]{\left\lfloor #1 \right\rfloor}
\newcommand{\ceil}[1]{\left\lceil #1 \right\rceil}
\newcommand{\twobitgadget}{\square}

\newcommand{\wh}{\widehat}

\newtheorem{theorem}{Theorem}[section]

\newtheorem{corollary}[theorem]{Corollary}
\newtheorem{remark}[theorem]{Remark}
\newtheorem{lemma}[theorem]{Lemma}
\newtheorem{claim}[theorem]{Claim}
\newtheorem{definition}[theorem]{Definition}
\newtheorem{question}[theorem]{Question}
\newtheorem{conjecture}[theorem]{Conjecture}

\newtheorem{fact}[theorem]{Fact}
\newtheorem{observation}[theorem]{Observation}

\DeclareMathOperator*{\Exp}{\mathbb{E}}

\newcommand{\zone}{\{0,1\}}
\newcommand{\pmone}{\{-1,1\}}

\newcommand{\R}{\mathbb{R}}

\newcommand{\OR}{\mathsf{OR}}
\newcommand{\NOR}{\mathsf{NOR}}
\newcommand{\AND}{\mathsf{AND}}
\newcommand{\XOR}{\mathsf{XOR}}
\newcommand{\PARITY}{\mathsf{PARITY}}
\newcommand{\IP}{\mathsf{IP}}
\newcommand{\EQ}{\mathsf{EQ}}
\newcommand{\wtcirc}{~\widetilde{\circ}~}
\newcommand{\DISJ}{\mathsf{DISJ}}

\newcommand{\QQ}{\mathsf{Q}}
\newcommand{\adeg}{\widetilde{\textnormal{deg}}}

\newcommand{\disc}{\textnormal{disc}}
\newcommand{\bdisc}{\textnormal{bdisc}}
\newcommand{\sign}{\text{sign}}

\newcommand{\UPP}{\mathsf{UPP}}

\newcommand{\ket}[1]{|#1\rangle}

\newcommand{\ip}[2]{\left\langle #1 | #2 \right\rangle}
\newcommand{\ketbra}[2]{|#1\rangle\! \langle #2|}

\newcommand{\eps}{\varepsilon}

\def\01{\{0,1\}}

\newcommand{\DeJo}{\mathsf{DJ}}
\newcommand{\asnorm}[2]{\|\widehat{#1}\|_{1,#2}}

\begin{document}

\title{The Role of Symmetry in\\
Quantum Query-to-Communication Simulation\thanks{This paper combines and extends results from two separate preliminary versions: one that appeared in the \href{https://doi.org/10.4230/LIPIcs.CCC.2020.32}{Proceedings of the 35th Computational Complexity Conference, 2020}~\cite{CCMP20} and another that appeared in the \href{https://doi.org/10.4230/LIPIcs.STACS.2022.20}{Proceedings of the 39th International Symposium on Theoretical Aspects of Computer Science, 2022}~\cite{CCH+20}.}}

\author{Sourav Chakraborty\thanks{Indian Statistical Institute, Kolkata. \texttt{sourav@isical.ac.in}}
\and
Arkadev Chattopadhyay\thanks{TIFR, Mumbai. Partially supported by the MATRICS grant MTR/2019/001633 of the Science and Engineering Research Board, DST, India. \texttt{arkadev.c@tifr.res.in}}
\and
Peter H\o yer\thanks{Department of Computer Science, University of Calgary, Canada. \texttt{hoyer@ucalgary.ca}}
\and
Nikhil S.~Mande\thanks{University of Liverpool, UK. Part of this work was done while the author was a postdoc at Georgetown University, and at CWI and QuSoft, supported by the Dutch Research Council (NWO) through QuantERA project QuantAlgo 680-91-034 and by the Quantum Software Consortium programme (project number 024.003.037). \texttt{Nikhil.Mande@liverpool.ac.uk}}
\and
Manaswi Paraashar\thanks{Indian Institute of Technology, Hyderabad, India. Work done while at Indian Statistical Institute, Kolkata. \texttt{manaswi.isi@gmail.com}}
\and
Ronald de Wolf\thanks{QuSoft, CWI and University of Amsterdam, the Netherlands. Partially supported by ERC Consolidator Grant 615307-QPROGRESS (which ended February 2019), and by the Dutch Research Council (NWO/OCW), as part of the Quantum Software Consortium programme (project number 024.003.037), and through QuantERA ERA-NET Cofund project QuantAlgo (680-91-034). \texttt{rdewolf@cwi.nl}}}
\date{}

\maketitle

\begin{abstract}
Buhrman, Cleve and Wigderson (STOC'98) showed that for every Boolean function $f : \pmone^n \to \pmone$ and $G \in \{\AND_2, \XOR_2\}$, the bounded-error quantum communication complexity of the composed function $f \circ G$ equals $O(\QQ(f) \log n)$, where $\QQ(f)$ denotes the bounded-error quantum query complexity of $f$. 
This is achieved by Alice running the optimal quantum query algorithm for~$f$, using a round of $O(\log n)$ qubits of communication to implement each query.
This is in contrast with the classical setting, where it is easy to show that $\mathsf{R}^{cc}(f \circ G) \leq 2\mathsf{R}(f)$, where $\mathsf{R}^{cc}$ and $\mathsf{R}$ denote bounded-error communication and query complexity, respectively. We show that the $O(\log n)$ overhead is required for some functions in the quantum setting, and thus the BCW simulation is tight. We note here that prior to our work, the possibility of $\mathsf{Q}^{cc}(f \circ G) = O(\mathsf{Q}(f))$, for all $f$ and all $G \in \{\AND_2, \XOR_2\}$, had not been ruled out. More specifically, we show the following.
\begin{itemize}
    \item We show that the $\log n$ overhead is \emph{not} required when $f$ is symmetric (i.e., depends only on the Hamming weight of its input), generalizing a result of Aaronson and Ambainis for the Set-Disjointness function (Theory of Computing'05). This upper bound assumes a shared entangled state, though for most symmetric functions the assumed number of entangled qubits is less than the communication and hence could be part of the communication.
    \item In order to prove the above, we design an efficient distributed version of noisy amplitude amplification that allows us to prove the result when $f$ is the OR function.
    \item In view of our first result above, one may ask whether the $\log n$ overhead in the BCW simulation can be avoided 
    even when $f$ is transitive, which is a weaker notion of symmetry.
    We give a strong negative answer by showing that the $\log n$ overhead is still necessary for some transitive functions even when we allow the quantum communication protocol an error probability that can be arbitrarily close to~$1/2$ (this corresponds to the unbounded-error model of communication).
    \item We also give, among other things, a general recipe to construct functions for which the $\log n$ overhead is required in the BCW simulation in the bounded-error communication model, even if the parties are allowed to share an arbitrary prior entangled state for free.
\end{itemize}
\end{abstract}

\tableofcontents

\section{Introduction}
\subsection{Motivation and main results}

The classical model of communication complexity was introduced by Yao~\cite{Yao79}, who also subsequently introduced its quantum analogue~\cite{Yao93}. 
Communication complexity has important applications in several disciplines, in particular for lower bounds on circuits, data structures, streaming algorithms, and many other complexity measures (see, for example, \cite{kushilevitz&nisan:cc} and the references therein).

A natural way to derive a communication problem from a Boolean function $f: \pmone^n \to \pmone$ is via composition. Let $f: \pmone^n \to \pmone$ be a function and let $G: \pmone^j \times \pmone^k \to \pmone$ be a ``two-party function.''
Then $F = f \circ G : \pmone^{nj} \times \pmone^{nk} \to \pmone$ denotes the function corresponding to the communication problem in which Alice is given input $X = (X_1, \dots, X_n) \in \pmone^{nj}$, Bob is given $Y = (Y_1, \dots, Y_n) \in \pmone^{nk}$, and their task is to compute $F(X, Y) = f(G(X_1, Y_1), \dots, G(X_n, Y_n))$. 
Many well-known functions in communication complexity are derived in this way, such as Set-Disjointness ($\DISJ_n := \NOR_n \circ \AND_2$), Inner Product ($\IP_{n} : = \PARITY_n \circ \AND_2$) and Equality ($\EQ_{n} : = \NOR_n \circ \XOR_2$). 
A natural approach to obtain efficient quantum communication protocols for $f \circ G$ is to ``simulate'' a quantum query algorithm for~$f$, where a query to the $i$th input bit of $f$ is simulated by a communication protocol that computes $G(X_i,Y_i)$.
Buhrman, Cleve and Wigderson~\cite{BuhrmanCleveWigderson98} observed that such a simulation is indeed possible if $G$ is $\AND_2$ or $\XOR_2$.

\begin{theorem}[\cite{BuhrmanCleveWigderson98}]
\label{theo: qtm_sim_log_loss}
For every Boolean function $f : \pmone^n \to \pmone$ and $\twobitgadget\in\{\AND_2,\XOR_2\}$, we have
\[
\QQ^{cc}\left(f \circ \twobitgadget \right) = O\left(\QQ(f) \log n\right).
\]
\end{theorem}

Here $\QQ(f)$ denotes the bounded-error quantum query complexity of $f$, and $\QQ^{cc}(f\circ \twobitgadget)$ denotes the bounded-error quantum communication complexity for computing $f\circ \twobitgadget$. Throughout this paper, we refer to Theorem~\ref{theo: qtm_sim_log_loss} as the BCW simulation.
\cite{BuhrmanCleveWigderson98} used this, for instance, to show that the bounded-error quantum communication complexity of the Set-Disjointness function is $O(\sqrt{n}\log n)$, using Grover's $O(\sqrt{n})$-query search algorithm~\cite{Grover96} for the $\NOR_n$ function.

It is folklore in the classical world that the analogous simulation does not incur a $\log n$ factor overhead. This is because when simulating a classical decision tree (or a randomly chosen decision tree) for $f$ to get a communication protocol for $f \circ \twobitgadget$, each node of the tree (which is of the form $z_i = \twobitgadget(x_i, y_i)$) can be evaluated by communicating two bits when $\twobitgadget \in \{\AND_2, \XOR_2\}$. That is,
\[
\mathsf{R}^{cc}\left(f \circ \twobitgadget\right) \leq 2\mathsf{R}(f), 
\]
where $\mathsf{R}(f)$ denotes the bounded-error randomized query complexity of $f$ and $\mathsf{R}^{cc}(f\circ \twobitgadget)$ denotes the bounded-error randomized communication complexity for computing $f\circ \twobitgadget$. Thus, a natural question is whether the multiplicative $\log n$ blow-up in the communication cost in the BCW simulation is necessary. 

For partial functions, tightness of the BCW simulation is known in some settings.
For example, consider the Deutsch-Jozsa ($\DeJo$) problem, where the input is an $n$-bit string with the promise that its Hamming weight is either 0 or $n/2$, and $\DeJo$ outputs $-1$ if the Hamming weight is $n/2$, and 1 otherwise.  $\DeJo$ has quantum query complexity $1$ whereas the \emph{exact} quantum communication complexity of $(\DeJo\circ \oplus)$ is $\log n$.  Note that it is unclear whether the $\log n$ factor loss here is additive or multiplicative.\footnote{Indeed, there are well-known situations where complexity of 1 vs.~$\log n$ can be deceptive. The classical private-coin randomized communication complexity of Equality is $\Theta(\log n)$, whereas the public-coin cost is well known to be $O(1)$. Newman's Theorem~\cite{newman:random} shows that this difference in costs, in general, is \emph{not multiplicative} but merely \emph{additive} (see~\cite{mande&wolf:equality} for a recent small improvement).}
Montanaro, Nishimura and Raymond~\cite{MNR11} exhibited a partial function for which the BCW simulation is tight (up to constants) in the \emph{exact} and \emph{non-deterministic} quantum settings.  They also observed the existence of a total function for which the BCW simulation is tight (up to constants) in the unbounded-error setting.
If we allow multi-output partial functions, then tightness of the BCW simulation is known: Consider the following function: $f$ takes an $n$-bit string $x$ as input which is promised to be a Hadamard codeword (see Definition~\ref{defi: hadamard codewords}), and outputs a $\log n$ bit string $z$ for which $x$ is its Hadamard codeword.  The communication problem $f \circ \wedge$, where the inputs to the two players are promised to be such that their bitwise-$\AND$ yields a Hadamard codeword, has bounded-error quantum communication complexity $\log n$, but $\QQ(f) = O(1)$.  Again, it is not clear here whether the $\log n$ factor loss is additive or multiplicative.

For the canonical problem of Set-Disjointness, Aaronson and Ambainis~\cite{aaronson&ambainis:searchj} (improving upon~\cite{HdW02}) showed that the $\log n$ overhead in the BCW simulation can be avoided. Since the outer function $\mathsf{NOR}_n$ is symmetric (i.e., it only depends on the Hamming weight of its input, its number of $-1$s), a natural question is whether the $\log n$ overhead can be avoided whenever the outer function is symmetric. Our first result gives a positive answer to this question.

\begin{theorem}
\label{theo:symmetric no logn intro}
For every symmetric Boolean function $f:\pmone^n\to\pmone$ and two-party function $G:\pmone^j\times\pmone^k\to\01$, we have
\begin{align*}
    \QQ^{cc,*}(f \circ G)=O(\QQ(f)\QQ_E^{cc}(G)).
\end{align*}
\end{theorem}

Here $\QQ^{cc,*}(F)$ denotes the bounded-error quantum communication complexity of two-party function $F$ when Alice and Bob shared an entangled state at the start of the protocol for free. $\QQ_E^{cc}(G)$ denotes the \emph{exact} quantum communication complexity of $G$, where the error probability is~0. In particular, if $G\in\{\AND_2,\XOR_2\}$ then $\QQ_E^{cc}(G)=1$ and hence $\QQ^{cc,*}(f \circ G)=O(\QQ(f))$.

\begin{remark}
If $\QQ(f)=\Theta(\sqrt{tn})$, then our protocol in the proof of Theorem~\ref{theo:symmetric no logn intro} starts from a shared entangled state of $O(t\log n)$ EPR-pairs. Note that if $t\leq n\QQ_E^{cc}(G)^2/(\log n)^2$ (this condition holds for instance if $\QQ_E^{cc}(G)\geq\log n$) then this number of EPR-pairs is no more than the amount of communication and hence might as well be established in the first message, giving asymptotically the same upper bound $\QQ^{cc}(f \circ G)=O(\QQ(f)\QQ_E^{cc}(G))$ for the model without prior entanglement. 
\end{remark}

The next question one might ask is whether one can weaken the notion of symmetry required in Theorem~\ref{theo:symmetric no logn intro}. A natural generalization of the class of symmetric functions is the class of \emph{transitive-symmetric} functions. A function $f : \pmone^n \to \pmone$ is said to be transitive-symmetric if for all $i, j \in [n]$, there exists $\sigma \in S_n$ such that $\sigma(i) = j$, and $f(x) = f(\sigma(x))$ for all $x \in \pmone^n$. Henceforth we refer to transitive-symmetric functions as simply transitive functions.
Can the $\log n$ overhead in the BCW simulation be avoided whenever the outer function is transitive?
We give a negative answer to this question in a strong sense:  the $\log n$ overhead is still necessary
even when we allow the quantum communication protocol  an error probability that can be arbitrarily close to~$1/2$.
We note here that prior to our work, even the possibility of $\mathsf{Q}^{cc}(f \circ G) = O(\mathsf{Q}(f))$, for all $f$ and all $G \in \{\AND_2, \XOR_2\}$, had not been ruled out.

\begin{theorem}
\label{theo: transitive upp lower bound intro}
There exists a transitive and total function $f : \pmone^{n} \to \pmone$, such that
    \begin{align*}
    \UPP^{cc}(f \circ \twobitgadget) & = \Omega(\QQ(f) \log n)
    \end{align*}
for every $\twobitgadget\in\{\AND_2,\XOR_2\}$.
\end{theorem}
Here $\UPP^{cc}(f \circ \twobitgadget)$ denotes the unbounded-error quantum communication complexity of $f \circ \twobitgadget$ (adding ``quantum'' here only changes the communication complexity by a constant factor~\cite{INRY07}).
The unbounded-error model of communication was introduced by Paturi and Simon~\cite{PS86} and is the strongest communication complexity model against which we know how to prove explicit lower bounds.
This model is known to be strictly stronger than the bounded-error quantum model. For instance, the Set-Disjointness function on $n$ inputs requires $\Omega(n)$ bits or $\Omega(\sqrt{n})$ qubits of  communication in the bounded-error model, but only requires $O(\log n)$ bits of communication in the unbounded-error model.
In fact, it follows from a recent result of Hatami, Hosseini and Lovett~\cite{HHL20} that there exists a function $F : \pmone^n\times\pmone^n \to \pmone$ with $\QQ^{cc,*}(F) = \Omega(n)$ while $\UPP^{cc}(F) = O(1)$.

A slightly relevant work to ours is that of Zhang~\cite{Zhang09}, who showed that for all total Boolean functions $f$, there must exist gadgets $g_i$, each either $\wedge$ or $\vee$, such that $\QQ^{cc}(f(g_1, \dots, g_n)) = \Omega(\textnormal{poly} (\QQ(f)))$.  For monotone $f$, they showed that either $\QQ^{cc}(f \circ \wedge) = \Omega(\textnormal{poly}(\QQ(f)))$ or $\QQ^{cc}(f \circ \OR_2) = \Omega(\textnormal{poly}(\QQ(f)))$.  They also state that it is unclear how tight the BCW simulation is. While Zhang's focus was on showing quantum communication lower bounds for composed functions of the form $f \circ g$ in terms of the hardness of $f$, our focus in this work is the necessity of the $O(\log n)$ overhead in the BCW simulation. Theorem~\ref{theo: transitive upp lower bound intro} shows that there exists a function for which it is tight up to constants, even in the \emph{unbounded-error} communication model.

Theorem~\ref{theo:symmetric no logn intro} and Theorem~\ref{theo: transitive upp lower bound intro} clearly demonstrate the role of symmetry in determining the presence of the $\log n$ overhead in the BCW query-to-communication simulation: 
this overhead is absent for symmetric functions (Theorem~\ref{theo:symmetric no logn intro}), but  present for a transitive function even when the model of communication under consideration is as strong as the unbounded-error model (Theorem~\ref{theo: transitive upp lower bound intro}).
We also give a general recipe to construct functions for which the $\log n$ overhead is required in the BCW simulation in the bounded-error communication model (see Theorem~\ref{theo: intro recipe for constructing BCW tight functions}).

\subsection{Overview of our approach and techniques}
In this section we discuss the ideas that go into the proofs of Theorem~\ref{theo:symmetric no logn intro} and Theorem~\ref{theo: transitive upp lower bound intro}.

\subsubsection{Communication complexity upper bound for symmetric functions}

To prove Theorem~\ref{theo:symmetric no logn intro}
we use the well-known fact that every symmetric function $f$ has an interval around Hamming weight $n/2$ where the function is constant; for $\mathsf{NOR}_n$ the length of this interval would be essentially $n$, while for $\PARITY_n$ it would be~1.
To compute $f$, it suffices to either determine that the Hamming weight of the input lies in that interval (because the function value is the same throughout that interval) or to count the Hamming weight exactly.

For two-party functions of the form $f\circ G$, we want to do this type of counting on the $n$-bit string $z=(G(X_1,Y_1),\ldots,G(X_n,Y_n))\in\pmone^n$. We show how this can be done with  $O(\QQ(f)\,\QQ^{cc}_E(G))$ qubits of communication if we had a quantum protocol that can find $-1$s in the string $z$ at a cost of $O(\sqrt{n}\,\QQ^{cc}_E(G))$ qubits. Such a protocol was already given by Aaronson and Ambainis for the special case where $G=\AND_2$ for their optimal quantum protocol for  Set-Disjointness, as a corollary of their quantum walk algorithm for search on a grid~\cite{aaronson&ambainis:searchj}. In this paper we give an alternative $O(\sqrt{n}\,\QQ^{cc}_E(G))$-qubit protocol. This implies the result of Aaronson and Ambainis as a special case, but it is arguably simpler and may be of independent interest. 

Our protocol can be viewed as an efficient distributed implementation of amplitude amplification with faulty components.
In particular, we replace the usual reflection about the uniform superposition by an imperfect reflection about the $n$-dimensional maximally entangled state ($=\log n$ EPR-pairs if $n$ is a power of~2). Such a reflection would require $O(\log n)$ qubits of communication to implement perfectly, but can be implemented with small error using only $O(1)$ qubits of communication, by invoking the efficient protocol of Aharonov et al.~\cite[Theorem~1]{AHLNSV14} that tests whether a given bipartite state equals the $n$-dimensional maximally entangled state. Still, at the start of this protocol we need to assume (or establish by means of quantum communication) a shared state of $\log n$ EPR-pairs. If $\QQ(f)=\Theta(\sqrt{tn})$ then our protocol for $f\circ G$ will run the $-1$-finding protocol $O(t)$ times, which accounts for our assumption that we share $O(t\log n)$ EPR-pairs at the start of the protocol.

\subsubsection{Communication complexity lower bound for transitive functions}
\label{sec: intro cc lb for transitive functions}
For proving Theorem~\ref{theo: transitive upp lower bound intro}, we exhibit a transitive function $f : \pmone^{2n^2} \to \pmone$ whose bounded-error quantum query complexity is $O(n)$ while the unbounded-error communication complexity of $f \circ \twobitgadget$ is $\Omega(n \log n)$ for $\twobitgadget \in \{\AND_2,\XOR_2\}$.

\textbf{Function construction and transitivity:} For the construction of $f$ we first require the definition of Hadamard codewords. The Hadamard codeword of $s \in \pmone^{\log n}$, denoted by $H(s) \in \pmone^n$, is a list of all parities of $s$. 
See Figure~\ref{fig: cexample} for a graphical visualization of $f$.

\begin{figure}
\begin{center}
\begin{tikzpicture}[scale=1]

\tikzstyle{gate}=[ellipse,draw=black]
\tikzstyle{input}=[]
\node(hdef) at (-6, -1){$f = $};
\node[gate] (root) at (0,0) {$\mathsf{PARITY}$};
\node[gate] (c1) at (-3,-2) {$h_{\IP_{\log n}}$};
\node[gate] (cn) at (3,-2) {$h_{\IP_{\log n}}$};
\node (dummy1) at (-2, -0.8) {};
\node (dummy2) at (2.1, -0.8) {$n$};
\draw[->, dashed] (dummy1) edge [bend right = 20] (dummy2);

\node (dummy3) at (-5, -2.7) {$2n$};
\node (dummy4) at (-1, -2.8) {};
\draw[<-, dashed] (dummy3) edge [bend right = 20] (dummy4);

\node[input] (x8) at (1,-2) {$\boldsymbol{\cdot}$};
\node[input] (x19) at (-1,-2) {$\boldsymbol{\cdot}$};
\node[input] (x419) at (0,-2) {$\boldsymbol{\cdot}$};

\node[input] (x11) at (-5.5,-4) {$x_{11}$};
\node[input] (x123) at (-4.5,-4) {$\cdots$};
\node[input] (x1f) at (-3.5,-4) {$x_{1n}$};
\node[input] (y11) at (-2.5,-4) {$y_{11}$};
\node[input] (y123) at (-1.5,-4) {$\cdots$};
\node[input] (y1f) at (-0.5,-4) {$y_{1n}$};

\node[input] (yn1) at (5.5,-4) {$y_{nn}$};
\node[input] (x13) at (4.5,-4) {$\cdots$};
\node[input] (ynf) at (3.5,-4) {$y_{n1}$};
\node[input] (xn1) at (2.5,-4) {$x_{nn}$};
\node[input] (y13) at (1.5,-4) {$\cdots$};
\node[input] (xnf) at (0.5,-4) {$x_{n1}$};

\draw[->] (c1) -- (root);
\draw[->] (cn) -- (root);

\draw[<-] (c1) -- (y11);
\draw[<-] (c1) -- (y1f);
\draw[<-] (c1) -- (x11);
\draw[<-] (c1) -- (x1f);

\draw[<-] (cn) -- (yn1);
\draw[<-] (cn) -- (ynf);
\draw[<-] (cn) -- (xn1);
\draw[<-] (cn) -- (xnf);

\draw[->] (root) -- ++(0,1);

\end{tikzpicture}
\end{center}
\caption{For a $(\log n)$-bit string $s$, let $H(s)$ denote the $n$-bit string that is a list of parities of all subsets of $s$. We refer to $H(s)$ as the \emph{Hadamard codeword corresponding to $s$} (see Definition~\ref{defi: hadamard codewords} for a formal definition). The following is a description of how $f$ behaves on an $2n^2$-bit input:
\\
If the inputs to the $j$-th $h_{\IP_{\log n}}$ are the Hadamard codewords $H(s_j)$ and $H(t_j)$ for all $j \in [n]$ and some $s_j, t_j \in \pmone^{\log n}$, then $f = \PARITY(\IP_{\log n}(s_1,t_1), \dots, \IP_{\log n}(s_n, t_n))$.  If there exists at least one $j \in [n]$ for which either $x_{j1}, \dots , x_{jn}$ or $y_{j1}, \dots, y_{jn}$ is not a Hadamard codeword, then $f$ outputs $-1$.}
\label{fig: cexample}
\end{figure}

Using properties of $\IP$ and Hadamard codewords, and the symmetry of $\PARITY_n$, we are able to show that $f$ is transitive (see Claim~\ref{clm:func_is_transitive}).

\textbf{Query upper bound:} The query upper bound of $O(n)$ is inspired by a query upper bound due to Ambainis and de Wolf~\cite{AdW14}. The following is a sketch of our upper bound.
\begin{itemize}
        \item Run $2n$ instances of the Bernstein-Vazirani algorithm~\cite{BV97}, two on each block, to obtain $2n$ strings of $\log n$ bits each. This algorithm guarantees that if all inputs to all of the $n$ inner gadgets $h_{\IP_{\log n}}$ were Hadamard codewords, then we would decode each Hadamard codeword correctly. 
        In this case, we can compute the output of each gadget correctly, using just $2n$ queries. Thus, if the inputs to all of the inner gadgets were indeed Hadamard codewords, then we could evaluate the function correctly with no further queries, since we know the whole input.

        \item In the next step, we first construct a $(2n^2)$-bit string using the output of the previous step (this string is the concatenation of the Hadamard codewords of the $2n$ many $(\log n)$-bit strings obtained). We then run Grover's search~\cite{Grover96, BHMT02} to check the equality of the input with this string; this tests whether the output of the first step was correct.  If it was correct, we succeed with probability 1.  If it was not correct, Grover's search catches a difference with probability at least $2/3$ and we output $-1$, succeeding with probability at least $2/3$ in this case.

        \item The $2n$ invocations of the Bernstein-Vazirani algorithm use a total of $2n$ queries, Grover's search uses another $O(n)$ queries, for a cumulative total of $O(n)$ queries.
    \end{itemize}
See the proof of Theorem~\ref{theo: query complexity doesn't increase on composition with hadamardization} for the query algorithm and its formal analysis.

\textbf{Communication lower bound:} 
Towards the unbounded-error communication lower bound, we first recall that each input block of $f$ equals $\IP_{\log n}$ if the inputs to each block are promised to be Hadamard codewords. Hence $f$ equals $\IP_{n \log n}$ under this promise, since $\PARITY_n \circ \IP_{\log n} = \IP_{n \log n}$. Thus by setting certain inputs to Alice and Bob suitably, $f \circ \twobitgadget$ is at least as hard as $\IP_{n \log n}$ for $\twobitgadget \in \{\AND_2, \XOR_2\}$ (for a formal statement, see Lemma~\ref{lm: reduction of communication problem from hadamardization} with $r = \PARITY_n$ and $g = \IP_{\log n}$). It is known from a seminal result of Forster~\cite{For02} that the unbounded-error communication complexity of $\IP_{n \log n}$ equals $\Omega(n \log n)$, completing the proof of the lower bound.

\subsection{Other results}
\begin{itemize}

    \item We give a general recipe for constructing a class of functions that witness tightness of the BCW simulation where the inner gadget is either $\AND_2$ or $\XOR_2$. However, the communication lower bound we obtain here is in the bounded-error model in contrast to Theorem~\ref{theo: transitive upp lower bound intro}, where the communication lower bound is proven in the unbounded-error model.

    The functions $f$ constructed for this purpose are composed functions similar to the construction in Figure~\ref{fig: cexample}, except that we are able to use a more general class of functions in place of the outer $\PARITY$ function, and also a more general class of functions in place of the inner $\IP_{\log n}$ functions.
    See Figure~\ref{fig: general_example} and its caption for an illustration and a more precise definition. 
    
    \begin{figure}
\begin{center}
\begin{tikzpicture}[scale=1]

\tikzstyle{gate}=[ellipse,draw=black]
\tikzstyle{input}=[]
\node(hdef) at (-6, -1){$f = $};
\node[gate] (root) at (0,0) {$r$};
\node[gate] (c1) at (-3,-2) {$h_{G}$};
\node[gate] (cn) at (3,-2) {$h_{G}$};
\node (dummy1) at (-2, -0.8) {};
\node (dummy2) at (2.1, -0.8) {$n$};
\draw[->, dashed] (dummy1) edge [bend right = 20] (dummy2);

\node (dummy3) at (-5, -2.7) {$2n$};
\node (dummy4) at (-1, -2.8) {};
\draw[<-, dashed] (dummy3) edge [bend right = 20] (dummy4);

\node[input] (x8) at (1,-2) {$\boldsymbol{\cdot}$};
\node[input] (x19) at (-1,-2) {$\boldsymbol{\cdot}$};
\node[input] (x419) at (0,-2) {$\boldsymbol{\cdot}$};

\node[input] (x11) at (-5.5,-4) {$x_{11}$};
\node[input] (x123) at (-4.5,-4) {$\cdots$};
\node[input] (x1f) at (-3.5,-4) {$x_{1n}$};
\node[input] (y11) at (-2.5,-4) {$y_{11}$};
\node[input] (y123) at (-1.5,-4) {$\cdots$};
\node[input] (y1f) at (-0.5,-4) {$y_{1n}$};

\node[input] (yn1) at (5.5,-4) {$y_{nn}$};
\node[input] (x13) at (4.5,-4) {$\cdots$};
\node[input] (ynf) at (3.5,-4) {$y_{n1}$};
\node[input] (xn1) at (2.5,-4) {$x_{nn}$};
\node[input] (y13) at (1.5,-4) {$\cdots$};
\node[input] (xnf) at (0.5,-4) {$x_{n1}$};

\draw[->] (c1) -- (root);
\draw[->] (cn) -- (root);

\draw[<-] (c1) -- (y11);
\draw[<-] (c1) -- (y1f);
\draw[<-] (c1) -- (x11);
\draw[<-] (c1) -- (x1f);

\draw[<-] (cn) -- (yn1);
\draw[<-] (cn) -- (ynf);
\draw[<-] (cn) -- (xn1);
\draw[<-] (cn) -- (xnf);

\draw[->] (root) -- ++(0,1);

\end{tikzpicture}
\end{center}
\caption{In this figure, $G:\pmone^{\log } \times \pmone^{\log n} \to \pmone$. If the inputs to the $j$-th $h_{G}$ are Hadamard codewords in $\pm H(s_j)$ and $\pm H(t_j)$ for all $j \in [n]$ and some $s_j, t_j \in \pmone^{\log n}$, then $f = r(G(s_1,t_1), \dots, G(s_n, t_n))$.  If there exists at least one $j \in [n]$ for which either $x_{j1}, \dots , x_{jn}$ or $y_{j1}, \dots, y_{jn}$ is not a Hadamard codeword, then $f$ outputs $-1$.}
\label{fig: general_example}
\end{figure}

    We require some additional constraints on the outer and inner functions. First, the approximate degree (see Definition~\ref{defi: adeg}) of the outer function, $r$, should be $\Omega(n)$.
    Second, the discrepancy of $G$ should be small with respect to some ``balanced'' probability distribution (see Definition~\ref{defi: discrepancy} and Definition~\ref{defi: balanced discrepancy} for formal definitions of these notions).

\begin{theorem}[Informal version of Theorem~\ref{theo: recipe for constructing BCW tight functions}]
\label{theo: intro recipe for constructing BCW tight functions}
    Let $r : \pmone^n \to \pmone$ be such that $\adeg(r) = \Omega(n)$ and let $G : \pmone^{\log n} \times \pmone^{\log n} \to \pmone$ be a total function. Define $f : \pmone^{2n^{2}} \to \pmone$ as in Figure~\ref{fig: general_example}.
    If there exists $\mu: \pmone^{\log n} \times \pmone^{\log n} \to \R$ that is a balanced probability distribution with respect to $G$ and $\disc_{\mu}(G) = n^{-\Omega(1)}$, then
    \begin{align*}
    \QQ(f) & = O(n),\\
    \QQ^{cc, *}(f \circ \twobitgadget) & = \Omega(n \log n),
    \end{align*}
    for every $\twobitgadget\in\{\AND_2,\XOR_2\}$.
\end{theorem}
    
    The query upper bound follows along similar lines as that of Theorem~\ref{theo: transitive upp lower bound intro}. For the lower bound, we first show via a reduction that for $f$ as described in Figure~\ref{fig: general_example} and $\twobitgadget \in \{\AND_2, \XOR_2\}$, the communication problem $f \circ \twobitgadget$ is at least as hard as $r \circ G$ (see Lemma~\ref{lm: reduction of communication problem from hadamardization}).
    This part of the lower bound proof is the same as in the proof of Theorem~\ref{theo: transitive upp lower bound intro}. For the hardness of $r \circ G$ (which in the case of Theorem~\ref{theo: transitive upp lower bound intro} turned out to be $\IP_{n \log n}$, for which Forster's theorem yields an unbounded-error communication lower bound), we are able to use a theorem implicit in a work of Lee and Zhang~\cite{LZ10}.
    This theorem gives a lower bound on the bounded-error communication complexity of $r \circ G$ in terms of the approximate degree of $r$ and the discrepancy of $G$ under a balanced distribution. For completeness, we provide an explicit proof in Appendix~\ref{sec: appendix communication complexity lower bound via generalized discrepancy method}.
    
    \item For a Boolean function $f : \pmone^n \to \pmone$, let $\log \|\wh{f}\|_{1, 1/3}$ denote the log-approximate-spectral norm of $f$ (see Definition~\ref{defn: awt}), and $\adeg_{1/3}(f)$ denote its approximate degree. 
    As an approach towards proving the Fourier entropy-influence conjecture, Arunachalam et al.~\cite{ACK+18} asked whether it is true for all total functions $f : \pmone^n \to \pmone$ that $\log \|\wh{f}\|_{1, 1/3} = O(\adeg_{1/3}(f))$. 
    We observe that this holds true when $f$ is a symmetric function.
    However, we give a negative answer to this question by showing that for the function $f$ used to prove Theorem~\ref{theo: transitive upp lower bound intro}, which is a \emph{transitive} function, $\log \|\wh{f}\|_{1, 1/3} = \Omega(\adeg_{1/3}(f) \log n)$.

    These are discussed in Section~\ref{sec: Appendix log-approximate-spectral norm and approximate degree transitive function}.
\end{itemize}

\subsection{Organization}
Section~\ref{sec: notation and prelims} introduces some notation and preliminaries. 
In Section~\ref{sec:noisyamplamplif} we construct our new one-sided error protocol for finding solutions in the string $z=(G(X_1,Y_1),\ldots,G(X_n,Y_n))\in\pmone^n$, as a corollary of our distributed version of amplitude amplification.
In Section~\ref{sec: No log-factor needed for symmetric functions} we prove Theorem~\ref{theo:symmetric no logn intro}, which shows that the $\log n$ overhead in the BCW simulation can be avoided when the outer function is symmetric; this relies on the protocol from Section~\ref{sec:noisyamplamplif}. 
In Section~\ref{sec: proofs of log n required} we prove Theorem~\ref{theo: transitive upp lower bound intro} and Theorem~\ref{theo: intro recipe for constructing BCW tight functions}, which are our results regarding necessity of the $\log n$ overhead in the BCW simulation in the unbounded-error and the bounded-error models of communication, respectively. As mentioned before, our proof of Theorem~\ref{theo: intro recipe for constructing BCW tight functions} requires a lower bound on the bounded-error quantum communication complexity of $f \circ G$ in terms of the approximate degree of $f$ and the discrepancy of $G$ under balanced distributions. Such a result is implicit in~\cite[Theorem 7]{LZ10}, but we provide a proof in Appendix~\ref{sec: appendix communication complexity lower bound via generalized discrepancy method}.

In Section~\ref{sec: Appendix log-approximate-spectral norm and approximate degree transitive function} we exhibit a transitive function $f$ on $n$ bits for which $\log(\|\wh{f}\|_{1, 1/3}) = \Omega(\adeg(f) \log n)$.

\section{Notation and preliminaries}
\label{sec: notation and prelims}
Without loss of generality, we assume $n$ to be a power of $2$ in this paper, unless explicitly stated otherwise. All logarithms in this paper are base 2.
Let $S_n$ denote the symmetric group over the set $[n]=\{1,\ldots,n\}$. For a string $x \in \pmone^n$ and $\sigma \in S_n$, let $\sigma(x)$ denote the string $x_{\sigma(1)}, \dots, x_{\sigma(n)} \in \pmone^n$. Consider an arbitrary but fixed bijection between subsets of $[\log n]$ and elements of $[n]$. For a string $s \in \pmone^{\log n}$, we abuse notation and also use $s$ to denote the equivalent element of $[n]$. The view we take will be clear from context. 
For a string $x \in \pmone^n$ and set $S \subseteq [n]$, define the string $x_S \in \pmone^{S}$ to be the restriction of $x$ to the coordinates in $S$.
Let $1^n$  and $(-1)^n$ denote the $n$-bit strings $(1, 1, \dots, 1)$  and $(-1, -1, \dots, -1)$, respectively. 

\subsection{Boolean functions}

For strings $x, y \in \pmone^n$, let $\langle x, y \rangle$ denote the inner product (mod 2) of $x$ and $y$. That is,
\[
\langle x, y \rangle = \prod_{i = 1}^n (x_i \wedge y_i).
\]
For every positive integer $n$, let $\PARITY_n : \{-1,1\}^n \to \{-1,1\}$ be defined as:
\begin{align*}
    \PARITY_n(x_1, \dots, x_n) = \prod_{i \in [n]} x_i.
\end{align*}

\begin{definition}[Symmetric functions]
A function $f : \pmone^n \to \pmone$ is symmetric if for all $\sigma \in S_n$ and for all $x \in \pmone^n$ we have $f(x) = f(\sigma(x))$.
\end{definition}

\begin{definition}[Transitive functions]
A function $f : \pmone^n \to \pmone$ is transitive if for all $i,j \in [n]$ there exists a permutation $\sigma \in S_n$ such that
\begin{itemize}
    \item $\sigma(i) = j$, and
    \item $f(x) = f(\sigma(x))$ for all $x \in \pmone^n$.
\end{itemize}
\end{definition}

We next discuss function composition. For total functions $f, g$, let $f \circ g$ denote the standard composition of the functions $f$ and $g$. 
We also require the following notion of composition of a total function $f$ with a partial function $g$.

\begin{definition}[Composition with partial functions]
\label{defi: composition with partial}
Let $f : \pmone^{n} \to \pmone$ be a total function and let $g : \pmone^{m} \to \cbra{-1, 1, \star}$ be a partial function. Let $f \wtcirc g : \pmone^{nm} \to \pmone$ denote the total function that is defined as follows on input $(X_1, \dots, X_n) \in \pmone^{nm}$, where $X_i \in \pmone^m$ for all $i \in [n]$.
\[
f \wtcirc g (X_1, \dots, X_n) = \begin{cases}
f(g(X_1), \dots, g(X_n)) & \text{if}~g(X_i) \in \pmone~\text{for all}~i \in [n],\\
-1 & \text{otherwise}.
\end{cases}
\]
\end{definition}
That is, we use $f \wtcirc g$ to denote the total function that equals $f \circ g$ on inputs when each copy of $g$ outputs a value in $\pmone$, and equals $-1$ otherwise.

\begin{definition}[Approximate degree]
\label{defi: adeg}
For every $\epsilon \geq 0$, the $\epsilon$-approximate degree of a function $f : \pmone^n \to \pmone$ is defined to be the minimum degree of a real polynomial $p : \pmone^n \to \R$ that uniformly approximates $f$ to error $\epsilon$. That is,
\[
\adeg_{\epsilon}(f) = \min\cbra{\deg(p) : |p(x) - f(x)| \leq \epsilon~\text{for all}~x \in \pmone^n}.
\]
Unless specified otherwise, we drop $\epsilon$ from the subscript and assume $\epsilon = 1/3$.
\end{definition}

We assume familiarity with quantum computing~\cite{nielsen&chuang:qc}, and
use $\QQ_{\epsilon}(f)$ to denote the $\epsilon$-error query complexity of $f$. Unless specified otherwise, we drop $\epsilon$ from the subscript and assume $\epsilon = 1/3$.

\begin{theorem}[\cite{bbcmw:polynomialsj}]
\label{theo: bbc+01}
Let $f: \pmone^n \to \pmone$ be a function. Then $\QQ(f) \geq \adeg(f)/2$.
\end{theorem}

\subsection{Fourier analysis of Boolean functions and multilinear polynomials}
\label{section: Fourier analysis of Boolean functions}
We define some basic notions from Fourier analysis on the Boolean cube.

Consider the vector space of functions from $\pmone^n$ to $\R$, equipped with an inner product defined by 
\[
\langle f, g \rangle := \Exp_{x \in \pmone^n}[f(x)g(x)] = \frac{1}{2^n}\sum_{x \in \pmone^n}f(x)g(x)
\]
for every $f,g: \pmone^n \rightarrow \R$. For any set $S \subseteq [n]$, define the associated \emph{parity} function $\chi_S$ by $\chi_S(x) = \prod_{i \in S}x_i$. The set $\cbra{\chi_S : S \subseteq [n]}$ of parity functions forms an orthonormal basis for this vector space.  Thus, every function $f : \pmone^n \to \R$ has a unique multilinear expression as $f = \sum_{S \subseteq [n]}\wh{f}(S)\chi_S$.
The coefficients $\cbra{\wh{f}(S) : S \subseteq [n]}$ are called the \emph{Fourier coefficients} of $f$.
The orthonormality of the $\chi_S$ functions easily implies the following fact.

\begin{fact}[Plancherel's Theorem]
\label{fact: plancherel's thm}
Let $f,g: \pmone^n \to \R$ be functions. Then,
\begin{align*}
    \langle f, g \rangle = \sum_{S \subseteq[n]} \wh{f}(S) \wh{g}(S).
\end{align*}
\end{fact}

We refer the reader to~\cite{odonnell:analysis} for more details on Fourier analysis of Boolean functions. More generally, we also consider the following class of functions whose domain is $\R^n$. 

\begin{definition}[Multilinear Polynomial]
A function $\phi : \R^n \to \R$ is a multilinear polynomial if $\phi$ is of the form:
$$
\phi(x_1, \dots, x_n) = \sum_{S \subseteq [n]}a_S \prod_{i \in S}x_i
$$
where $a_S \in \R$.
\end{definition}

\begin{definition}[Symmetric Multilinear Polynomial]
\label{def: Symmetric Multilinear Polynomial}
A multilinear polynomial $\phi : \R^n \to \R$ is said to be symmetric if $\phi(x_1, \dots, x_n) = \phi(x_{\sigma(1)}, \dots, x_{\sigma(n)})$ for all $(x_1, \dots, x_n) \in \R^n$ and $\sigma \in S_n$.
\end{definition}

\begin{definition}[Spectral Norm of a Multilinear Polynomial]
Let $\phi : \R^n \to \R$ be a multilinear polynomial of the form $\phi(x_1, \dots, x_n) = \sum_{S \subseteq [n]}a_S \prod_{i \in S}x_i$. The spectral norm of $\phi$, denoted by $\|\phi\|_1$, is defined as  
$$
\|\phi\|_1 = \sum_{S \subseteq [n]} |a_S|.
$$
\end{definition}

\begin{definition}[Approximate Spectral Norm]\label{defn: awt}
The \emph{approximate spectral norm} of a function  $f : \pmone^n \to \{-1,1\}$, denoted by $\asnorm{f}{\epsilon}$ is defined to be the minimum spectral norm of a real polynomial $p : \pmone^n \to \R$ that satisfies $\abs{p(x) - f(x)} \leq \eps$ for all $x \in \pmone^n$ for which $f(x) \in \pmone$.
\[
\asnorm{f}{\epsilon} := \min\cbra{\|{p}\|_1 : \abs{p(x) - f(x)} \leq \epsilon~\text{for all~} x \in \pmone^n~\text{for which}~f(x) \in \pmone}.
\]
\end{definition}

\begin{fact}[Properties of Spectral Norm of Multilinear Polynomials]
\label{fact:l1_norm_facts}
Let $f,g : \R^n \to \R$ be any symmetric polynomials and let $\alpha \in \R$ be any real number. Then,
\begin{enumerate}
    \item $\|\alpha f\|_1 = |\alpha| \|f\|_1$,
    \item $\|f + g\|_1 \leq \|f\|_1 + \|g\|_1$,
    \item $\|fg\|_1 \leq \|f\|_1\|g\|_1$.
\end{enumerate}
\end{fact}

\subsection{Communication complexity}

We assume familiarity with communication complexity~\cite{kushilevitz&nisan:cc}. 

\begin{definition}[Two-party function]
\label{defi: two party functions}
We call a function $G: \pmone^j \times \pmone^k \to \pmone$ a \emph{two-party function} to indicate that it corresponds to a communication problem in which Alice is given input $x \in \pmone^j$, Bob is given input $y \in \pmone^k$, and their task is to compute $G(x, y)$.
\end{definition}

Here, and in the rest of the paper, the two parties Alice and Bob must agree upon the output of a communication protocol.

\begin{definition}[Composition with two-party functions]
\label{defi: comm problem from two party functions}
Let $f: \pmone^n \to \pmone$ be a function and let $G: \pmone^j \times \pmone^k \to \pmone$ be a two-party function. Then $F = f \circ G : \pmone^{nj} \times \pmone^{nk} \to \pmone$ denotes the two-party function corresponding to the communication problem in which Alice is given input $X = (X_1, \dots, X_n) \in \pmone^{nj}$, Bob is given $Y = (Y_1, \dots, Y_n) \in \pmone^{nk}$, and their task is compute $F(X, Y) = f(G(X_1, Y_1), \dots, G(X_n, Y_n))$.
\end{definition}

\begin{definition}[Inner Product function]
\label{defi: inner product function}
For every positive integer $n$, define the function $\IP_n : \pmone^{n} \times \pmone^{n} \to \pmone$ by 
\[
\IP_n(x,y) = \langle x,y \rangle.
\]
In other words, $\IP_n = \PARITY_n \circ \AND_2$.
\end{definition}

\begin{observation}
\label{obs: parity composed inner product}
For all positive integers $k, t$, $\PARITY_k \circ \IP_t = \IP_{kt}$.
\end{observation}

We also assume familiarity with quantum communication complexity~\cite{wolf:qccsurvey}.
We use $\QQ_{\epsilon}^{cc}(G)$ and  $\QQ_{\epsilon}^{cc,*}(G)$ to represent the $\epsilon$-error quantum communication complexity of a two-party function $G$ in the models without and with unlimited shared entanglement, respectively. The latter means that Alice and Bob share at the start of the protocol an entangled state of their choice (independent of their inputs $X$ and $Y$) at no cost, for instance a large number of EPR-pairs. One can think of this as the quantum generalization of shared randomness. Unless specified otherwise, we drop $\epsilon$ from the subscript and assume $\epsilon = 1/3$.

\begin{definition}[Balanced probability distribution]
\label{defi: balanced probability distribution}
We call a probability distribution $\mu: \pmone^n \to \R$ balanced with respect to a function $f: \pmone^n \to \pmone$ if $\sum_{x \in \pmone^n} f(x) \mu(x) = 0$. 
\end{definition}

\begin{definition}[Discrepancy]
\label{defi: discrepancy}
Let $G : \pmone^j \times \pmone^k \to \pmone$ be a function and $\lambda$ be a distribution on $\pmone^j \times \pmone^k$. For every $S \subseteq \pmone^j$ and $T \subseteq \pmone^k$, define
\[
\disc_\lambda(S \times T, G) = \abs{\sum_{x, y \in S \times T}G(x, y)\lambda(x, y)}.
\]
The discrepancy of $G$ under the distribution $\lambda$ is defined to be
\[
\disc_{\lambda}(G) = \max_{S \subseteq \pmone^j, T \subseteq \pmone^k} \disc_{\lambda}(S \times T, G),
\]
and the discrepancy of $f$ is defined to be
\[
\disc(G) = \min_{\lambda}\disc_{\lambda}(G).
\]
\end{definition}

\begin{definition}[Balanced-discrepancy]
\label{defi: balanced discrepancy}
Let $G : \pmone^j \times \pmone^k \to \pmone$ be a function and $\Lambda$ be the set of all balanced distributions on $\pmone^j \times \pmone^k$ with respect to $G$.
The balanced-discrepancy of $G$ is defined to be
\[
\bdisc(G) = \min_{\lambda \in \Lambda}\disc_{\lambda}(G).
\]
\end{definition}

The following theorem is implicit in~\cite[Theorem 7]{LZ10}. We prove it in Appendix~\ref{sec: appendix communication complexity lower bound via generalized discrepancy method} for completeness, extending ideas from~\cite{Cha09}.

\begin{theorem}
\label{theo:  discrepancy lb on quantum communication}
Let $r : \pmone^n \to \pmone$ and $G : \pmone^j \times \pmone^k \to \pmone$ be functions.
Let $\mu: \pmone^j \times \pmone^k \to \R$ be a balanced distribution with respect to $G$ and $\disc_{\mu}(G) = o(1)$. If $\frac{8en}{\adeg(r)}\leq \left(\frac{1}{\disc_{\mu}(G)}\right)^{1-\beta}$ for some constant $\beta \in (0,1)$, then
\begin{align*}
    \QQ^{cc, *}(r \circ G) =  \Omega\left(\adeg(r) \log\left(\frac{1}{\disc_{\mu}(G)}\right)\right).
\end{align*}
In particular,
\begin{align*}
    \QQ^{cc, *}(r \circ G) =  \Omega\left(\adeg(r) \log\left(\frac{1}{\bdisc(G)}\right)\right).
\end{align*}
\end{theorem}

\subsection{Hadamard encoding}
Recall that we index coordinates of $n$-bit strings by integers in $[n]$, and also interchangeably by strings in $\pmone^{\log n}$ via the natural correspondence. For $x \in \pmone^n$, let $-x \in \pmone^n$ be defined as $(-x)_i = -x_i$ for all $i \in [n]$. We use the notation $\pm x$ to denote the set $\{x, -x\}$.

\begin{definition}[Hadamard Codewords]
\label{defi: hadamard codewords}
For every positive integer $n$ and $s \in \pmone^{\log n}$, let $H(s) \in \pmone^n$ be defined as
\begin{align*}
    (H(s))_t = \prod_{i: s_i = -1} t_i~\text{for all}~t \in \pmone^{\log n}.
\end{align*}
If $x \in \pmone^n$ is such that $x = H(s)$ for some $s \in \pmone^{\log n}$, we say $x$ is a Hadamard codeword corresponding to $s$.
\end{definition}
That is, for every $s \in \pmone^{\log n}$, there is an $n$-bit Hadamard codeword corresponding to $s$. This represents the enumeration of all parities of $s$.

We now define how to encode a two-party total function $G$ on $(\log j + \log k)$ input bits to a partial function $h_G$ on $(j+k)$ input bits, using Hadamard encoding.

\begin{definition}[Hadamardization of functions]
\label{defi: hadamardization}
Let $j, k \geq 1$ be powers of 2, and let $G : \pmone^{\log j} \times \pmone^{\log k} \to \pmone$ be a function. Define a partial function $h_G : \pmone^{j + k} \to \cbra{-1, 1, \star}$ by
\begin{align*}
h_G(x, y) = \begin{cases}
G(s, t) & \text{if}~x \in \pm H(s), y \in \pm H(t)~\text{for some}~s \in \pmone^{\log j}, t \in \pmone^{\log k}\\
\star & \text{otherwise}.
\end{cases}
\end{align*}
\end{definition}

\subsection{Additional concepts from quantum computing}\label{sec:quantumdefs}
The Bernstein-Vazirani algorithm~\cite{BV97} is a quantum query algorithm that takes an $n$-bit string as input and outputs a $(\log n)$-bit string. The algorithm has the following properties:
\begin{itemize}
    \item the algorithm makes one quantum query to the input and
    
    \item if the input $x \in \pmone^n$ satisfies $x\in\pm H(s)$ for some $s \in \pmone^{\log n}$, then the algorithm returns $s$ with probability $1$.
\end{itemize}

Consider a symmetric Boolean function $f:\pmone^n\to\pmone$. Define the quantity 
\[
\Gamma(f)=\min\{|2k-n+1| : f(x)\neq f(y)\text{ if }|x|=k\text{ and } |y|=k+1\}
\]
from~\cite{paturi:degree}. 
One can think of $\Gamma(f)$ as essentially the length of the interval of Hamming weights around $n/2$ where $f$ is constant (for example, for the majority and parity functions this would be~1, and for $\OR_n$ this would be $n-1$). It is known that this quantity determines the bounded-error quantum query complexity of~$f$:

\begin{theorem}[{\cite[Theorem~4.10]{bbcmw:polynomialsj}, see \cite{wolf:degreesymmf} for a tight $\eps$-dependent bound}]
\label{th:bbcmwsymmetricupperbound}
For every symmetric function $f : \pmone^n \to \pmone$, we have
\[
\QQ(f) = \Theta(\sqrt{(n-\Gamma(f))n}).
\]
\end{theorem}
The upper bound follows from a quantum algorithm that exactly counts the Hamming weight $|x|$ of the input if $|x|\leq t$ or $|x|\geq n-t$ for $t=\ceil{(n-\Gamma(f))/2}$, and that otherwise learns $|x|$ is in the interval $[t+1,n-t-1]$ (which is an interval around $n/2$ where $f(x)$ is constant). By the definition of $\Gamma(f)$, this information about $|x|$ suffices to compute~$f(x)$. In Section~\ref{sec: No log-factor needed for symmetric functions} we use this observation to give an efficient quantum communication protocol for a two-party function $f\circ G$. 

We will need a unitary protocol that allows Alice and Bob to implement an approximate reflection about the $n$-dimensional maximally entangled state \[
\ket{\psi}=\frac{1}{\sqrt{n}}\sum_{i\in\01^{\log n}}\ket{i}\ket{i}.
\]
Ideally, such a reflection would map $\ket{\psi}$ to itself, and put a minus sign in front of all states orthogonal to $\ket{\psi}$.
Doing this perfectly would requires $O(\log n)$ qubits of communication.
Fortunately we can derive a cheaper protocol from a test that Aharonov et al.~\cite[Theorem~1]{AHLNSV14} designed, which uses $O(\log(1/\eps))$ qubits of communication and checks whether a given bipartite state equals $\ket{\psi}$, with one-sided error probability $\eps$. By the usual trick of running this protocol, applying a $Z$-gate to the answer qubit, and then reversing the protocol, we can implement the desired reflection approximately.\footnote{Possibly with some auxiliary qubits on Alice and Bob's side which start in $\ket{0}$ and end in $\ket{0}$, except in a part of the final state that has norm at most $\eps$.}

In a bit more detail, we can modify the protocol given in~\cite[Proof of Theorem~1]{AHLNSV14} as follows. Let register $a$ start with $\log(d)$ $\ket{0}$-qubits, with Alice mapping this to the uniform superposition $\ket{U}=1/\sqrt{d}\sum_{i=1}^d\ket{i}$ in step 2. Steps 3-5 are performed as in~\cite[Proof of Theorem~1]{AHLNSV14}. In step 6, instead of applying a measurement Bob applies the unitary that reflects about the uniform superposition (i.e., it puts a minus sign in front of all states orthogonal to $\ket{U}$). We then reverse steps 1-5 in order to set the auxiliary qubits of register $a$ back to $\ket{0}$. If the~\cite{AHLNSV14}-protocol (including the measurement) produces the correct output bit with success probability $\geq 1-\eps$, then our modified procedure unitarily implements a reflection about $\ket{\psi}$ within operator-norm error $O(\sqrt{\eps})$. Note that we will have some auxiliary qubits (the register~$a$) which start and end in $\ket{0}$, and which, for simplicity, we will not write when we invoke this result.

We thus have the following theorem.

\begin{theorem}\label{th:reflectmaxent}
Let $R_{\psi}=2\ketbra{\psi}{\psi}-I$ be the reflection about the maximally entangled state shared between Alice and Bob.
There exists a protocol that uses $O(\log(1/\eps))$ qubits of communication and that implements a unitary $R^\eps_{\psi}$ such that $\norm{R^\eps_{\psi} - R_{\psi}}\leq\eps$
and $R^\eps_{\psi}\ket{\psi}=\ket{\psi}$.
\end{theorem}

We use $\UPP^{cc}(F)$ to denote unbounded-error quantum communication complexity of two-party function~$F$. 
It is folklore (see for example~\cite{INRY07}) that the unbounded-error quantum communication complexity\footnote{The unbounded-error model does not allow shared randomness or prior shared entanglement (which yields shared randomness by measuring) between Alice and Bob, since any two-party function~$F$ would have constant communication complexity in that setting, see, for example,~\cite[Section 1.2]{INRY07}.}
of~$F$ equals its classical counterpart up to a factor of at most $2$, so it does not really matter much whether we use $\UPP^{cc}$ for classical unbounded-error communication complexity (as it is commonly used) or for quantum unbounded-error complexity. Crucially, for both the complexity of $\IP_n$ is linear in $n$:

\begin{theorem}[\cite{For02}]
\label{theo: forster}
Let $n$ be a positive integer. Then,
\[
\UPP^{cc}(\IP_n) = \Omega(n).
\]
\end{theorem}

\section{Noisy amplitude amplification and a new distributed-search protocol}\label{sec:noisyamplamplif}
In this section we present a version of quantum amplitude amplification that still works if the reflections involved are not perfectly implemented.
In particular, the usual reflection about the uniform superposition will be replaced in the communication setting by an imperfect reflection about the $n$-dimensional maximally entangled state, based on the communication-efficient protocol of Aharonov et al.~\cite[Theorem~1]{AHLNSV14} for testing whether Alice and Bob share that state.
This allows us to avoid the $\log n$ factor that would be incurred if we instead used a BCW-style distributed implementation of standard amplitude amplification, with $O(\log n)$ qubits of communication to implement each query.
Our main result in this section is the following general theorem, which allows us to search among a sequence of two-party instances $(X_1,Y_1),\ldots,(X_n,Y_n)$ for an index $i \in [n]$ where $G(X_i,Y_i)=-1$, for any two-party function~$G$.

\begin{theorem}\label{th:searchG}
Let $G: \pmone^j \times \pmone^k \to \pmone$ be a two-party function, $X = (X_1, \dots, X_n) \in \pmone^{nj}$ and $Y = (Y_1, \dots, Y_n) \in \pmone^{nk}$.
Define $z=(G(X_1,Y_1),\ldots,G(X_n,Y_n))\in\pmone^n$.
Assume Alice and Bob start with $\ceil{\log n}$ shared EPR-pairs. 
\begin{itemize}
\item There exists a quantum protocol using $O(\sqrt{n}\,\QQ_E^{cc}(G))$ qubits of communication that finds (with success probability $\geq 0.99$) an $i\in [n]$ such that $z_i=-1$ if such an $i$ exists, and says ``no'' with probability~1 if no such~$i$ exists.
\item If the number of $-1$s in $z$ is within a factor of 2 from a known integer~$t$, then the communication can be reduced to $O(\sqrt{n/t}\,\QQ_E^{cc}(G))$ qubits.
\end{itemize}
\end{theorem}

\begin{remark}
 The $\log n$ shared EPR-pairs that we assume Alice and Bob share at the start could also be established by means of $\log n$ qubits of communication at the start of the protocol.
 For the result in the first bullet, this additional communication does not change the asymptotic bound.
For the result of the second bullet,  if $t\leq n\QQ_E^{cc}(G)^2/(\log n)^2$, then this additional communication does not change the asymptotic bound either. However, if $t=\omega(n/(\log n)^2)$ and $\QQ_E^{cc}(G)=O(1)$ then the quantum communication $O(\sqrt{n/t}\,\QQ_E^{cc}(G))$ is $o(\log n)$ and establishing the $\log n$ EPR-pairs by means of a first message makes a difference.
 \end{remark}

As a corollary, we obtain a new $O(\sqrt{n})$-qubit protocol for the distributed search problem composed with $G=\AND_2$ (whose decision version is the Set-Disjointness problem).

\subsection{Amplitude amplification with perfect reflections}

We first describe basic amplitude amplification (with perfect reflections) in a slightly unusual recursive manner, similar to~\cite{hmw:berrorsearch}. We are dealing with a search problem where some set $\cal G$ of basis states are deemed ``good'' and the other basis states are deemed ``bad.'' Let $P_{\cal G}=\sum_{g\in{\cal G}}\ketbra{g}{g}$ be the projector onto the span of the good basis states, and $O_{\cal G}=I-2P_{\cal G}$ be the reflection that puts a `$-$' in front of the good basis states: $O_{\cal G}\ket{g}=-\ket{g}$ for all basis states $g\in{\cal G}$, and $O_{\cal G}\ket{b}=\ket{b}$ for all basis states $b\not\in{\cal G}$.

Suppose we have an initial state $\ket{\psi}$ which is a superposition of a good state and a bad state:
$$
\ket{\psi}=\sin(\theta)\ket{G}+\cos(\theta)\ket{B},
$$
where $\ket{G}=P_{\cal G}\ket{\psi}/\norm{P_{\cal G}\ket{\psi}}$ and $\ket{B}=(I-P_{\cal G})\ket{\psi}/\norm{(I-P_{\cal G})\ket{\psi}}$. For example in Grover's algorithm, with a search space of size $n$ containing $t$ solutions, the initial state $\ket{\psi}$ would be the uniform superposition, and its overlap (inner product) with the good subspace spanned by the $t$ ``good'' (sometimes called ``marked'') basis states would be $\sin(\theta)=\sqrt{t/n}$.

We would like to increase the weight of the good state, increasing the initially-small angle $\theta$ closer to $\pi/2$.
Let $R_\psi$ denote the reflection about the state $\ket{\psi}$, i.e., $R_\psi\ket{\psi}=\ket{\psi}$ and $R_\psi{\ket{\phi}}=-\ket{\phi}$ for every $\ket{\phi}$ that is orthogonal to $\ket{\psi}$. Then the algorithm $A_1=R_\psi\cdot O_{\cal G}$ is the product of two reflections, which (in the 2-dimensional space spanned by $\ket{G}$ and $\ket{B}$) corresponds to a rotation by the angle $2\theta$, thus increasing our angle from $\theta$ to $3\theta$. This is the basic amplitude amplification step. It maps
$$
\ket{\psi}\mapsto A_1\ket{\psi}=\sin(3\theta)\ket{G}+\cos(3\theta)\ket{B}.
$$
We can now repeat this step recursively, defining 
$$
A_2=A_1R_\psi A^*_1\cdot O_{\cal G}\cdot A_1.
$$
Note that $A_1R_\psi A^*_1$ is a reflection about the state $A_1\ket{\psi}$. Thus $A_2$ triples the angle between $A_1\ket{\psi}$ and $\ket{B}$, mapping
$$
\ket{\psi}\mapsto A_2\ket{\psi}=\sin(9\theta)\ket{G}+\cos(9\theta)\ket{B}.
$$
Continuing recursively in this fashion, define the algorithm 
\begin{equation}\label{eq:recursivedefAk}
A_{j+1}=A_j R_\psi A_j^*\cdot O_{\cal G}\cdot A_j.
\end{equation}
The last algorithm $A_k$ will map
$$
\ket{\psi}\mapsto A_k\ket{\psi}=\sin(3^k\theta)\ket{G}+\cos(3^k\theta)\ket{B}.
$$
Hence after $k$ recursive amplitude amplification steps, we have angle $3^k\theta$. Since we want to end up with angle $\approx \pi/2$, if we know $\theta$ then we can choose 
\begin{equation}\label{eq: choiceofk}
k=\floor{\log_3(\pi/(2\theta))}.
\end{equation}
This gives us an angle $3^k\theta\in(\pi/6,\pi/2]$, so the final state $A_k\ket{\psi}$ has overlap $\sin(3^k\theta)>1/2$ with the good state $\ket{G}$.

Let $C_k$ denote the ``cost'' (in whatever measure,
for example query complexity, or communication complexity, or circuit size) of the algorithm~$A_k$. Looking at its recursive definition (Equation~\eqref{eq:recursivedefAk}), $C_k$ is 3 times $C_{k-1}$, plus the cost of $R_\psi$ plus the cost of $O_{\cal G}$. If we just count applications of $O_{\cal G}$ (``queries''), considering $R_\psi$ to be free, then $C_{k+1}=3C_k+1$. This recursion has the closed form $C_k=\sum_{i=0}^{k-1} 3^i< 3^k$.  With the above choice of $k$ we get $C_k=O(1/\theta)$.
In the case of Grover's algorithm, where $\theta=\arcsin(\sqrt{t/n})\approx\sqrt{t/n}$, the cost is $C_k=O(\sqrt{n/t})$.

\subsection{Amplitude amplification with imperfect reflections}

Now we consider the situation where we do not implement the reflections $R_\psi$ perfectly, but instead implement another unitary~$R^\eps_\psi$ at operator-norm distance~$\norm{R^\eps_\psi-R_\psi}\leq\eps$ from $R_\psi$, with the additional property that $R^\eps_\psi\ket{\psi}=\ket{\psi}$ (this one-sided error property will be important for the proof). 
We can control and vary this error~$\eps$, but smaller $\eps$ will typically correspond to higher cost of $R^\eps_\psi$. The reflection $O_{\cal G}=I-2P_{\cal G}$ will still be implemented perfectly below.

Recursively define the same algorithms as in the previous section, but now with imperfect $R_\psi$'s with errors $\eps_1,\eps_2,\ldots$ that we will specify later:
\begin{equation}\label{eq:defAjfaulty}
A_1=R^{\eps_1}_\psi\cdot O_{\cal G}\mbox{~~and~~}A_{j+1}=A_jR^{\eps_{j+1}}_\psi A_j^*\cdot O_{\cal G}\cdot A_j.
\end{equation}
The main new ingredient in the analysis below will be to have different good and bad states in each recursive round of amplitude amplification, not just the same $\ket{G}$ and $\ket{B}$ throughout.
We will start with initial state $\ket{\psi}=\ket{\psi_0}$, and for $j\geq 1$ define the state after the $j$th recursion as $\ket{\psi_j}=A_j\ket{\psi_0}$. For $j\geq 0$, define $\ket{G_j}=P_{\cal G}\ket{\psi_j}/\norm{P_{\cal G}\ket{\psi_j}}$ and $\ket{B_j}=(I-P_{\cal G})\ket{\psi_j}/\norm{(I-P_{\cal G})\ket{\psi_j}}$ as the normalized good and bad states after the $j$th recursion, and define angle $\theta_j$ such that
\[
\ket{\psi_j}=\sin(\theta_j)\ket{G_j}+\cos(\theta_j)\ket{B_j}.
\]
Our goal here is to show that the small initial angle $\theta_0$ grows to a constant~$\theta_k$ after some $k$ iterations, because then measuring $\ket{\psi_k}$ has $\Omega(1)$ probability to return a good basis state (if the constant is small, we can just repeat $O(1)$ times to find a good basis state with high success probability). 

Let us consider how $A_{j+1}$ acts on $\ket{\psi}$. First, $A_j$ maps $\ket{\psi}$ to $\ket{\psi_j}$. Second, $O_{\cal G}$ negates $\ket{G_j}$, mapping $\ket{\psi_j}$ to $\ket{\psi_j}-2\sin(\theta_j)\ket{G_j}$. Third, $A_jR^{\eps_{j+1}}_\psi A_j^*$ is applied.
If $R^{\eps_{j+1}}_\psi$ were equal to $R_\psi$, then this application would be a perfect reflection through the state $\ket{\psi_j}$, which would complete one perfect round of amplitude amplification on top of $A_j$, and we would obtain  $\sin(3\theta_j)\ket{G_j}+\cos(3\theta_j)\ket{B_j}$.
Now what happens if instead we use $R_\psi^{\eps_{j+1}}$? Because $R_\psi^{\eps_{j+1}}$ acts like $R_\psi$ on $\ket{\psi}$ (we have one-sided error), $A_j R_\psi^{\eps_{j+1}} A_j^*$ acts like the perfect reflection $A_j R_\psi A_j^*$ on $\ket{\psi_j}$, sending that state to itself. So the only error incurred by the use of the slightly-faulty $R_\psi^{\eps_{j+1}}$ is on $-2\sin(\theta_j)\ket{G_j}$. Hence we have
\begin{eqnarray*}
\ket{\psi_{j+1}} & = & A_{j+1}\ket{\psi}\\
& = & A_j R_\psi^{\eps_{j+1}} A_j^*~(\ket{\psi_j}-2\sin(\theta_j)\ket{G_j})\\ 
& = & A_j R_\psi A_j^*~(\ket{\psi_j}-2\sin(\theta_j)\ket{G_j})~~+~~A_j (R_\psi^{\eps_{j+1}}-R_\psi) A_j^*~(\ket{\psi_j}-2\sin(\theta_j)\ket{G_j})\\
& = & \sin(3\theta_j)\ket{G_j} + \cos(3\theta_j)\ket{B_j} 
~~-~~A_j (R_\psi^{\eps_{j+1}}-R_\psi) A_j^*~2\sin(\theta_j)\ket{G_j}.
\end{eqnarray*}
The norm of the last term  in the last line is $\leq 2 \eps_{j+1} \sin(\theta_j)$, because $\norm{R^{\eps_{j+1}}_\psi -R_\psi}\leq\eps_{j+1}$, and $A_j$ and $A_j^*$ are norm-preserving.

Expressing $\ket{\psi_{j+1}}$ as $\sin(\theta_{j+1})\ket{G_{j+1}} + \cos(\theta_{j+1})\ket{B_{j+1}}$,
we see that 
\[
\sin(\theta_{j+1}) = \norm{P_{\cal G}\ket{\psi_{j+1}}}= \sin(3\theta_j) \pm 2 \eps_{j+1}\sin(\theta_j),
\]
where ``$\pm\eta$'' means ``within $\eta$''.
Accordingly, as long as $\theta_j$ and $\eps_{j+1}$ are small, $\theta_{j+1}$ will be roughly $3\theta_j$, as it was in the case with perfect $R_\psi$. 

To make this precise, using the trigonometric identities $\sin(3x)=\sin(2x)\cos(x)+\sin(x)\cos(2x)$, $\sin(2x)=2\sin(x)\cos(x)$, and $\cos(2x)=1-2\sin(x)^2$, we can write
\[
\sin(3\theta_j)=3\sin(\theta_j)-4\sin(\theta_j)^3.
\]
Hence
\[
\frac{\sin(3\theta_j)}{\sin(\theta_j)}
=3-4\sin(\theta_j)^2
\]
and
\[
\frac{\sin(\theta_k)}{\sin(\theta_0)}=\prod_{j=0}^{k-1}\frac{\sin(\theta_{j+1})}{\sin(\theta_j)}
=\prod_{j=0}^{k-1}\left(\frac{\sin(3\theta_j)}{\sin(\theta_j)}\pm 2\eps_{j+1}\right)
=\prod_{j=0}^{k-1}\left(3-4\sin(\theta_j)^2\pm 2\eps_{j+1}\right).
\]
Our goal is to lower bound $\sin(\theta_k)$, showing that it grows as roughly $3^k\sin(\theta_0)$. Because $\sin(\theta_j)$ occurs negatively on the right-hand side, we first need to show that it cannot grow too fast: 
\[
\frac{\sin(\theta_j)}{\sin(\theta_0)}=\prod_{\ell=0}^{j-1}\frac{\sin(\theta_{\ell+1})}{\sin(\theta_\ell)}
\leq \prod_{\ell=0}^{j-1}\left(3+2\eps_{\ell+1}\right)= 3^j\prod_{\ell=0}^{j-1}\left(1+\frac{2}{3}\eps_{\ell+1}\right)\leq 3^j e^{(2/3)\sum_{\ell=1}^j \eps_\ell},
\]
where the last inequality used $1+x\leq e^x$.
Choosing 
\begin{equation}\label{eq: epsj}
\eps_j=\frac{1}{10\cdot 2^j}, 
\end{equation}
we have $\sum_{j=1}^k\eps_j \leq 1/10$, and $e^{(2/3)\sum_{\ell=1}^k\eps_\ell}\leq 1.2$.
We can upper bound
\[
\sum_{j=0}^{k-1} \sin(\theta_j)^2\leq 1.5\sum_{j=0}^{k-1}9^j\sin(\theta_0)^2\leq \frac{3}{4} 9^k\sin(\theta_0)^2 .
\]
We can now lower bound 
\begin{align*}
\frac{\sin(\theta_k)}{\sin(\theta_0)}
& =\prod_{j=0}^{k-1}\left(3-4\sin(\theta_j)^2\pm2\eps_{j+1}\right)
=3^k\prod_{j=0}^{k-1}\left(1-\frac{4}{3}\sin(\theta_j)^2\pm\frac{2}{3}\eps_{j+1}\right)\\
& \geq 3^k\left(1-\sum_{j=0}^{k-1}(\frac{4}{3}\sin(\theta_j)^2+\frac{2}{3}\eps_{j+1})\right)\geq 3^k\left(1-9^k\sin(\theta_0)^2-1/15\right),
\end{align*}
where the first inequality uses that $(1-a)(1-b)\geq 1-(a+b)$ for all $a,b\geq 0$.
We can now conclude that as long as $k\leq \log_3(1/\sin(\theta_0))-1$, we have the desired growth of the weight on the good state: 
$\sin(\theta_k)\geq  3^k\sin(\theta_0)/2$.

\paragraph{Overlap with the initial good state.}
The above argument shows that the final state $\ket{\psi_k}$ after $k=\log_3(1/\sin(\theta_0))-O(1)$ recursions has constant overlap with the subspace spanned by the good basis states: $\norm{P_{\cal G}\ket{\psi_k}}=\sin(\theta_k)\geq  3^k\sin(\theta_0)/2=\Omega(1)$. In principle, this final good state $\ket{G_k}$ could be very different from the initial good state~$\ket{G_0}$. For example, it might be that  $\ket{G_0}$ is a uniform superposition over some good basis states that have a desired extra property (as it will be in our application in the next subsection), while $\ket{G_k}$ is some weird non-uniform superposition over all good basis states. 
We will now look ``inside'' $\ket{G_j}$ and show that in fact the overlap $\gamma_j=\ip{G_0}{G_j}$ (which starts at $\gamma_0=1$) remains large throughout the run of the algorithm, meaning that the final state $\ket{\psi_k}$ actually has $\Omega(1)$ overlap with the initial good state~$\ket{G_0}$. This will used crucially in Section~\ref{sec: Distributed amplitude amplification with imperfect reflection}.

Write $\ket{G_j}=\gamma_j\ket{G_0}+\sqrt{1-|\gamma_j|^2}\ket{G_j'}$, 
where $\ket{G_j'}=(P_{\cal G}-\ketbra{G_0}{G_0})\ket{G_j}$ is the part of the good state $\ket{G_j}$ that is orthogonal to $\ket{G_0}$. 
If $\eps_{j+1}$ were 0, then $\sin(\theta_j)\ket{G_j}$ would be ``lifted'' to $\sin(3\theta_j)\ket{G_j}$ by $A_{j+1}$, and hence $\sin(\theta_j)\gamma_j\ket{G_0}$ would be ``lifted'' to $\sin(3\theta_j)\gamma_j\ket{G_0}$; $\ket{G_{j+1}}$ would be equal to $\ket{G_j}$, and $\gamma_{j+1}$ would be equal to~$\gamma_j$. However, because of the use of the faulty reflection~$R_\psi^{\eps_{j+1}}$, 
$\gamma_{j+1}$ can be smaller than $\gamma_j$, albeit not by very much: $\sin(\theta_{j+1})\gamma_{j+1}\geq \sin(3\theta_j)\gamma_j-2\eps_{j+1}\sin(\theta_j)$ and hence
\[
|\gamma_{j+1}|\geq \frac{|\sin(3\theta_j)\gamma_j-2\eps_{j+1}\sin(\theta_j)|}{\sin(\theta_{j+1})} 
\geq \gamma_j\frac{\sin(3\theta_j)}{\sin(\theta_{j+1})}-2\eps_{j+1}\frac{\sin(\theta_j)}{\sin(\theta_{j+1})}\geq
|\gamma_j|-4\eps_{j+1},
\]
where we used that $\sin(\theta_j)/\sin(\theta_{j+1})\leq 1$, which is true as long as $\theta_j\leq \pi/6$, and
\[
\frac{\sin(3\theta_j)}{\sin(\theta_{j+1})}
\geq\frac{\sin(3\theta_j)}{\sin(3\theta_j)+2\eps_{j+1}\sin(\theta_j)} = 1- \frac{2\eps_{j+1}\sin(\theta_j)}{\sin(3\theta_j)+2\eps_{j+1}\sin(\theta_j)}\geq 
1- 2\eps_{j+1}\frac{\sin(\theta_j)}{\sin(3\theta_j)}
\geq 1-2\eps_{j+1}.
\]
Because we start at $\gamma_0=1$, after $k$ iterations we still have $|\gamma_k|\geq 1-4\sum_{j=1}^k\eps_j\geq 0.6$.

\paragraph{The cost of $A_k$.}
Unwrapping the recursive definition of $A_k$ of Equation~\eqref{eq:defAjfaulty}, note that it uses 
\begin{itemize}
\item 1 application of $R_\psi^{\eps_k}$
\item 3 applications of $R_\psi^{\eps_{k-1}}$
\item 9 applications of $R_\psi^{\eps_{k-2}}$
\item \ldots
\item $3^{k-1}$ applications of $R_\psi^{\eps_1}$
\item $\sum_{j=0}^{k-1} 3^j=(3^k-1)/2$ applications of $O_{\cal G}$.
\end{itemize}
Since $R_\psi^{\eps_j}$ is applied much more often for smaller $j$, we want to make it cheaper for smaller $j$, which corresponds to allowing larger error for it as reflected in our choice $\eps_j=\frac{1}{10\cdot 2^j}$.

\subsection{Distributed amplitude amplification with imperfect reflection}
\label{sec: Distributed amplitude amplification with imperfect reflection}

We will now instantiate the above scheme to the case of \emph{distributed} search, where our measure of cost is communication, that is, the number of qubits sent between Alice and Bob. Specifically, consider the \emph{intersection problem} where Alice and Bob have inputs $x\in\pmone^n$ and $y\in\pmone^n$, respectively.
Assume for simplicity that $n$ is a power of~2, so $\log n$ is an integer. Alice and Bob want to find  an $i\in\{0,\ldots,n-1\}=\01^{\log n}$ such that $x_i=y_i=-1$, if such an $i$ exists.

The basis states in this distributed problem are $\ket{i}\ket{j}$, and we define the set of ``good'' basis states as
\[
{\cal G}=\{\ket{i}\ket{j}\mid x_i=y_j=-1\},
\]
even though we are only looking for $i,j$ with extra property that $i=j$ (it's easier to implement $O_{\cal G}$ with this more liberal definition of $\cal G$ that also allows $i\neq j$). Our protocol will start with the maximally entangled initial state $\ket{\psi}$ in $n$ dimensions, which corresponds to $\log n$ EPR-pairs: 
$$
\ket{\psi}=\ket{\psi_0}=\frac{1}{\sqrt{n}}\sum_{i\in\01^{\log n}}\ket{i}\ket{i}=\sin(\theta_0)\ket{G_0}+\cos(\theta)\ket{B_0},
$$
where we assume there are $t$ $i$'s where $x_i=y_i=-1$, i.e., $t$ solutions to the intersection problem, so
\begin{equation}\label{eq:thetatovern}
\theta_0=\arcsin(\sqrt{t/n})=\Omega(\sqrt{t/n}). 
\end{equation} 
and 
\[
\ket{G_0}=\frac{1}{\sqrt{t}}\sum_{i : \ket{i}\ket{i}\in{\cal G}}\ket{i}\ket{i}. 
\]
It costs $\ceil{\log n}$ qubits of communication between Alice and Bob to establish this initial shared state~$\ket{\psi}$, or it costs nothing if we assume pre-shared entanglement. 
Our goal is to end up with a state that has large inner product with $\ket{G_0}$.

In order to be able to use amplitude amplification, we would like to be able to reflect about the above state $\ket{\psi}$. However, in general this perfect reflection $R_{\psi}$ costs a lot of communication: Alice would send her $\log n$ qubits to Bob, who would unitarily put a $-1$ in front of all states orthogonal to $\ket{\psi}$, and then sends back Alice's qubits. This has a communication cost of $O(\log n)$ qubits, which is too much for our purposes.
Fortunately, Theorem~\ref{th:reflectmaxent} gives us a way to implement a one-sided $\eps$-error reflection protocol $R^\eps_{\psi}$ that only costs $O(\log(1/\eps))$ qubits of communication.

The reflection $O_{\cal G}$ puts a `$-$' in front of the basis states $\ket{i}\ket{j}$ in $\cal G$. This can be implemented perfectly using only 2 qubits of communication, as follows. To implement this reflection when she has basis state $\ket{i}$, Alice XORs $(1 - x_i)/2$ into a fresh auxiliary $\ket{0}$-qubit and sends this qubit to Bob. Bob receives this qubit and applies the following unitary map:
$$
\ket{b}\ket{j}\mapsto y_j^{b}\ket{b}\ket{j},~~~b\in\01,j\in[n].
$$ 
He sends back the auxiliary qubit. Alice sets the auxiliary qubit back to $\ket{0}$ by XOR-ing $(1 - x_i)/2$ into it. Ignoring the auxiliary qubit (which starts and ends in state $\ket{0}$), this maps $\ket{i}\ket{j}\mapsto (-1)^{[x_i=y_j=-1]}\ket{i}\ket{j}$. Hence we have implemented $O_{\cal G}$ correctly: a minus sign is applied exactly for the good basis states $\ket{i}\ket{j}$, the ones where $x_i=y_j=-1$.

Now consider the algorithms (more precisely, communication protocols) defined recursively exactly like in Equation~\eqref{eq:defAjfaulty}:
$$
A_1=R^{\eps_1}_\psi \cdot O_{\cal G}
\mbox{~~and~~}A_{j+1}=A_jR^{\eps_{j+1}}_\psi A_j^*\cdot O_{\cal G}\cdot A_j
$$
with the choice of $\eps_j$'s from Equation~\eqref{eq: epsj}. If we pick $k=\floor{\log_3(1/\sin(\theta_0))-1}$, 
then the previous subsection showed that there will be $\Omega(1)$ overlap between the bipartite state $\ket{\psi_k}$ that Alice and Bob have created, and the state $\ket{G_0}$.
At this point Alice and Bob can measure, and with probability $\Omega(1)$ they will each see the same $i$, with the property that $x_i=y_i=-1$.\footnote{Note that Alice and Bob can check very cheaply whether their respective measurement outcomes $i$ and $j$ satisfy this by running a classical equality protocol.} We can repeat this $O(1)$ times to ensure that the probability of finding an intersection is close to~1 (if there are indeed $t$ intersecting $i$'s). Observe that the same argument for correctness probability works if instead of knowing $t$ exactly, Alice and Bob know a number which is within factor $2$ of $t$. 

What about the communication complexity of $A_k$?
From Equation~\eqref{eq:defAjfaulty} and Theorem~\ref{th:reflectmaxent},
the recursion for the communication costs of these algorithms is 
$$
C_{j+1}=3C_j+O(\log(1/\eps_{j+1}))+2.
$$ 
Solving this recurrence with our $\eps_j$'s from Equation~\eqref{eq: epsj} and the value of $\theta_0$ from Equation~\eqref{eq:thetatovern} we obtain
$$
C_k=\sum_{j=1}^k 3^{k-j} (O(\log(1/\eps_j))+2)=\sum_{j=1}^k 3^{k-j}O(j)= O(3^k)=O(1/\sin(\theta_0))=O(\sqrt{n/t}).
$$
Thus, using $O(\sqrt{n/t})$ qubits of communication we can find (with constant success probability) an intersection point~$i$. This also allows us to solve the Set-Disjointness problem, which is the decision problem whose output is 1 if there is no intersection between $x$ and $y$. Recall from the above discussion that if the $t$ that we used equals the actual number of solutions only up to a factor of~2, then the above protocol still has $\Omega(1)$ probability to find a solution, and $O(1)$ repetitions will boost this success probability to~$0.99$. In case we do not even know $t$ approximately, we can use the standard technique of trying exponentially decreasing guesses for $t$ to find an intersection point with communication of $O(\sqrt{n})$ qubits. 

Note that there is no log-factor in the communication complexity, in contrast to the original $O(\sqrt{n}\log n)$-qubit Grover-based quantum protocol for the intersection problem of Buhrman et al.~\cite{BuhrmanCleveWigderson98}. Aaronson and Ambainis~\cite{aaronson&ambainis:searchj} earlier already managed to remove the log-factor, giving an $O(\sqrt{n})$-qubit protocol for Set-Disjointness as a consequence of their local version of quantum search on a grid graph ($O(\sqrt{n})$ qubits of communication is optimal~\cite{razborov:qdisj}). We have just reproved this result of~\cite{aaronson&ambainis:searchj} in a different and arguably simpler way.

The above description is geared towards the intersection problem, where the ``inner'' function is $G=\AND_2$: we called a basis state $\ket{i}\ket{j}$ ``good'' if $x_i=y_j=-1$.
However, this can easily be generalized to the situation where Alice and Bob's respective inputs are $X=(X_1,\ldots,X_n)$ and $Y=(Y_1,\ldots,Y_n)$ and we want to find an $i\in[n]$ where $G(X_i,Y_i)=-1$ for some
two-party function~$G$, and define the set of ``good'' basis states as ${\cal G}=\{\ket{i}\ket{j}\mid G(X_i,Y_j)=-1\}$.\footnote{We intentionally use the letter `$G$' to mean ``good'' in $\cal G$ and to refer to the two-party function $G$, since $G$ determines which basis states $\ket{i}\ket{j}$ are ``good.''} The only thing that changes in the above is the implementation of the reflection~$O_{\cal G}$, which would now be computed by means of an exact quantum communication protocol for $G(X_i,Y_j)$, at a cost of $2\QQ_E^{cc}(G)$ qubits of communication.\footnote{The factor of 2 is to reverse the protocol after the phase $G(X_i,Y_j)$ has been added to basis state $\ket{i}\ket{j}$, in order to set any workspace qubits back to~$\ket{0}$.}
Note that because we can check (at the expense of another $\QQ_E^{cc}(G)$ qubits of communication) whether the output index $i$ actually satisfies $G(X_i,Y_i)=-1$, we may assume the protocol has one-sided error: it always outputs ``no'' if there is no such $i$. This concludes the proof of  Theorem~\ref{th:searchG}.

\section{No log-factor needed for symmetric functions}
\label{sec: No log-factor needed for symmetric functions}

In this section we prove Theorem~\ref{theo:symmetric no logn intro} from the introduction.
Consider a symmetric Boolean function $f:\pmone^n\to\pmone$. As explained in Section~\ref{sec:quantumdefs}, there is an integer $t=\ceil{(n-\Gamma(f))/2}$ such that we can compute $f$ if we learn the Hamming weight $|z|$ of the input $z\in\pmone^n$ or learn that $|z|\in[t+1,n-t-1]$. The bounded-error quantum query complexity is $\QQ(f)=\Theta(\sqrt{tn})$ (Theorem~\ref{th:bbcmwsymmetricupperbound}).

For a given two-party function $G:\pmone^j\times\pmone^k\to\pmone$, we have an  induced two-party function $F:\pmone^{nj}\times\pmone^{nk}\to\pmone$ defined as
$$F(X_1,\ldots,X_n,Y_1,\ldots,Y_n)=f(G(X_1,Y_1),\ldots,G(X_n,Y_n)).$$ 
Define 
\[
z=(G(X_1,Y_1),\ldots,G(X_n,Y_n))\in\pmone^n. 
\]
Then $F(X,Y)=f(z)$ only depends on the number of $-1$s in $z$. The following theorem allows us to count this number using $O(\QQ(f)\,\QQ_E^{cc}(G))$ qubits of communication.

\begin{theorem}
For every $t$ between $1$ and $n/2$, there exists a quantum protocol that starts from $O(t\log n)$ EPR-pairs, communicates $O(\sqrt{tn}\,\QQ_E^{cc}(G))$ qubits, and tells us $|z|$ or tells us that $|z|>t$, with error probability $\leq 1/8$.
\end{theorem}

\begin{proof}
Abbreviate $q=\QQ_E^{cc}(G)$. Our protocol has two parts: the first filters out the case $|z|\geq 2t$, while the second finds all solutions if $|z|<2t$.

\paragraph{Part 1.}
First Alice and Bob decide between the case (1) $|z|\geq 2t$ and the case (2) $|z|\leq t$ (even though $|z|$ might also lie in $\{t+1,\ldots,2t-1\}$) using $O(\sqrt{n}q)$ qubits of communication, as follows. They use shared randomness to choose a uniformly random subset $S\subseteq[n]$ of $\ceil{n/(2t)}$ elements. Let $E$ be the event that $z_i=-1$ for at least one $i\in S$. By standard calculations there exist $p_1,p_2\in[0,1]$ with $p_1=p_2+\Omega(1)$ 
such that $\Pr[E]\geq p_1$ in case (1) and $\Pr[E]\leq p_2$ in case (2). Alice and Bob use the distributed-search protocol from the first bullet of Theorem~\ref{th:searchG} to decide $E$, with $O(\sqrt{|S|}\,q)=O(\sqrt{n}\,q)$ qubits of communication (plus a negligible $O(\log n)$ EPR-pairs) and error probability much smaller than $p_1-p_2$. By repeating this a sufficiently large constant number of times and seeing whether the fraction of successes was larger or smaller than $(p_1+p_2)/2$, they can distinguish between cases (1) and (2) with success probability $\geq 15/16$. If they conclude they're in case (1) then they output ``$|z|>t$'' and otherwise they proceed to the second part of the protocol. 

Note that if $|z|\in\{t+1,\ldots,2t-1\}$ (the ``grey zone'' in between cases (1) and (2)), then we can't give high-probability guarantees for one output or the other, but concluding (1) leads to the correct output ``$|z|>t$'' in this case, while concluding (2) means the protocol proceeds to Part~2. So either course of action is fine if $|z|\in\{t+1,\ldots,2t-1\}$.

By Newman's theorem~\cite{newman:random} the shared randomness used for choosing $S$ can be replaced by $O(\log n)$ bits of private randomness on Alice's part, which she can send to Bob in her first message, so Part~1 communicates $O(\sqrt{n}\,q)$ qubits in total.

\paragraph{Part 2.}
We condition on Part~1 successfully filtering out case (1), so from now on assume $|z|<2t$.
Our goal in this second part of the protocol is to find all indices $i$ such that $z_i=-1$ (we call such $i$ ``solutions''), with probability $\geq 15/16$, using $O(\sqrt{tn}\,q)$ qubits of communication. This will imply that the overall protocol is correct with probability $\geq 1-1/16-1/16=7/8$, and uses $O(\sqrt{tn}\,q)$ qubits of communication in total. For an integer $k\geq 1$, consider the following protocol $P_k$.

\begin{algorithm}[H]
\SetAlgoLined
\textbf{Input:} An integer $k \geq 1$ \\
\Repeat{$200\sqrt{2^k n}\,q$ qubits have been sent}{
    \begin{enumerate}
        \item Run the protocol from the last bullet  of Theorem~\ref{th:searchG} with $t=2^{k-1}$.\\ (suppressing some constant factors, assume for simplicity that this uses $\sqrt{n/2^k}\,q$ qubits of communication, $\log n$ shared EPR-pairs at the start, and has probability~$\geq 1/100$ to find a solution if the actual number of solutions is in $[t/2,2t]$).
        \item If the output of the protocol from Theorem~\ref{th:searchG} is $i$, then they replace $X_i,Y_i$ by some pre-agreed inputs $X'_i,Y'_i$, respectively, such that $G(X'_i,Y'_i)=1$ (this reduces the number of $-1$s in $z$ by~1)
    \end{enumerate}
}
\caption{Protocol $P_k$}
\end{algorithm}

\begin{claim}\label{claim:pkcostanderror}
Suppose $|z|\in[2^{k-1},2^k)$.
Then protocol $P_k$ uses $O(\sqrt{2^k n}\,q)$ qubits of communication, assumes $O(2^k\log n)$ EPR-pairs at the start of the protocol, and finds at least $|z|-2^{k-1}+1$ solutions, except with probability $\leq 1/2$.
\end{claim}

\begin{proof}
The upper bound on the communication is obvious from the stopping criterion of $P_k$.

As long as the remaining number of solutions is $\geq 2^{k-1}$, each run of the protocol has probability $\geq 1/100$ to find another solution. Hence the expected number of runs of the protocol of Theorem~\ref{th:searchG} to find at least $|z|-2^{k-1}+1$ solutions, is $\leq 100(|z|-2^{k-1}+1)$. By Markov's inequality, the probability that we haven't yet found $|z|-2^{k-1}+1$ solutions after $\leq 200(|z|-2^{k-1}+1)\leq 100\cdot 2^k$ runs, is $\leq 1/2$. The communication cost of so many runs is $O(2^k (\sqrt{n/2^k}\,q )) = O(\sqrt{2^k n}\,q)$ qubits.
Hence by the time that the number of qubits of the stopping criterion have been communicated, we have probability $\geq 1/2$ of having found at least $|z|-2^{k-1}+1$ solutions. The assumed number of EPR-pairs at the start is $\log n$ per run, so $O(2^k\log n)$ in total.
\end{proof}

Note that if we start with a number of solutions $|z|\in[2^{k-1},2^k)$, and $P_k$ succeeds in finding at least $|z|-2^{k-1}+1$ new solutions, then afterwards we have $< 2^{k-1}$ solutions left. The following protocol runs these $P_k$ in sequence, pushing down the remaining number of solutions to~0.

\begin{algorithm}[H]
\SetAlgoLined
\For{$k=\ceil{\log_2 (2t)}$ \textnormal{downto} $1$}{
\begin{enumerate}
        \item Run $P_k$ a total of $r_k=\ceil{\log_2 (2t)}-k+5$ times (replacing all $-1$s found by $+1s$ in $z$).
        
        \item Output the total number of solutions found.
    \end{enumerate}
}
\caption{Protocol $P$}
\end{algorithm}

\begin{claim}
If $|z|<2t$ then protocol $\cal P$ uses $O(\sqrt{tn}\,q)$ qubits of communication, assumes $O(t\log n)$ EPR-pairs at the start of the protocol, and outputs $|z|$, except with probability $\leq 1/16$.
\end{claim}

\begin{proof}
First, by Claim~\ref{claim:pkcostanderror}, the total number of qubits communicated is
\[
\sum_{k=1}^{\ceil{\log_2 (2t)}} r_k\cdot O(\sqrt{2^k n}\,q)
=O(\sqrt{tn}\,q)\cdot \sum_{\ell=0}^{\ceil{\log_2 (2t)}-1} (\ell+5) /\sqrt{2^\ell}
=O(\sqrt{tn}\,q),
\]
where we used a variable substitution $k=\ceil{\log_2 (2t)}-\ell$. Second, the number of EPR-pairs we're starting from is 
\[
\sum_{k=1}^{\ceil{\log_2 (2t)}} 
r_k\cdot O(2^k\log n)=O(t\log n)\cdot \sum_{\ell=0}^{\ceil{\log_2 (2t)}-1} 
(\ell+5) /2^\ell=O(t\log n).
\]
Third, by Claim~\ref{claim:pkcostanderror} and the fact that we are performing $r_k$ repetitions of $P_k$, if the $k$th round of $P$ starts with a remaining number of solutions that is in the interval $[2^{k-1},2^k)$ then that round ends with $<2^{k-1}$ remaining solutions, except with probability at most $1/2^{r_k}$. By the union bound, the probability that any one of the $\ceil{\log_2 (2t)}$ rounds does not succeed at this, is at most
\[
\sum_{k=1}^{\ceil{\log_2 (2t)}} \frac{1}{2^{r_k}}=\sum_{\ell=0}^{\ceil{\log_2 (2t)}-1}\frac{1}{2^{\ell+5}}\leq\frac{1}{16}.
\]
Since $2^{\ceil{\log_2 (2t)}}\geq 2t$ and we start with $|z|<2t$, if each round succeeds, then by the end of $\cal P$ there are no remaining solutions left. Thus, the protocol $\cal P$ finds all solutions and learns $|z|$ with probability $\geq 15/16$.
\end{proof}
Part~1 and Part~2 each have error probability $\leq 1/16$, so by the union bound the protocol succeeds except with probability $1/8$.
If $|z|\geq 2t$ then Part~1 outputs the correct answer  ``$|z|>t$''; if $|z|\leq t$ then all solutions (and hence $|z|$) are found by Part~2; and if $|z|\in\{t+1,\ldots,2t-1\}$ then either Part~1 already outputs the correct answer  ``$|z|>t$'' or the protocol proceeds to Part~2 which then finds all solutions.
\end{proof}

We can use the above theorem twice: once to count the number of $-1$s in $z$ (up to $t$) and once to count the number of $1$s in $z$ (up to $t$).
This uses $O(\sqrt{tn}\,\QQ^{cc}_E(G))=O(\QQ(f)\,\QQ_E^{cc}(G))$ qubits of communication, 
assumes $O(t\log n)$ shared EPR-pairs at the start of the protocol, and gives us enough information about $|z|$ to compute $f(z)=F(X,Y)$.
This concludes the proof of Theorem~\ref{theo:symmetric no logn intro} from the introduction, restated below.

\begin{theorem}[Restatement of Theorem~\ref{theo:symmetric no logn intro}]
For every symmetric Boolean function $f:\pmone^n\to\pmone$ and two-party function $G:\pmone^j\times\pmone^k\to\01$, we have
\begin{align*}
    \QQ^{cc,*}(f \circ G)=O(\QQ(f)\QQ_E^{cc}(G)).
\end{align*}
If $\QQ(f)=\Theta(\sqrt{tn})$, then our protocol assumes a shared state of $O(t\log n)$ EPR-pairs at the start.
\end{theorem}

We remark that for the special case where $G=\AND_2$, our upper bound matches the lower bound proved by Razborov~\cite{razborov:qdisj}, except for symmetric functions~$f$ where the first switch of function value happens at Hamming weights very close to~$n$.
In particular, if $f=\AND_n$ and $G=\AND_2$, then $\QQ^{cc}(f \circ G)=1$ but $\QQ(f)=\Theta(\sqrt{n})$.

\section{Necessity of the log-factor overhead in the BCW simulation}
\label{sec: proofs of log n required}

In this section we prove Theorem~\ref{theo: transitive upp lower bound intro} and Theorem~\ref{theo: intro recipe for constructing BCW tight functions}. 
For Theorem~\ref{theo: transitive upp lower bound intro} we exhibit a function $f : \pmone^{2n^2} \to \pmone$ for which $\QQ(f)= O(n)$ and $\UPP(f \circ \twobitgadget) = \Omega(n \log n)$ for $\twobitgadget\in\{\AND_2,\XOR_2\}$.
In Theorem~\ref{theo: intro recipe for constructing BCW tight functions} we show a general recipe for constructing total functions $f : \pmone^{2n^2} \to \pmone$ such that $\QQ(f) = O(n)$ and $\QQ^{cc, *}(f \circ \twobitgadget) = \Omega(n \log n)$ for $\twobitgadget\in\{\AND_2,\XOR_2\}$. We first give a formal statement of Theorem~\ref{theo: intro recipe for constructing BCW tight functions}.

\begin{theorem}[Restatement of Theorem~\ref{theo: intro recipe for constructing BCW tight functions}]
\label{theo: recipe for constructing BCW tight functions}
Let $r : \pmone^n \to \pmone$ be a total function with $\adeg(r) = \Omega(n)$, $G : \pmone^{\log n} \times \pmone^{\log n} \to \pmone$ be a total function and $\twobitgadget\in\{\AND_2,\XOR_2\}$. Define $f = r \wtcirc h_G : \pmone^{2n^2} \to \pmone$. If there exists $\mu: \pmone^{\log n} \times \pmone^{\log n} \to \R$ that is a balanced probability distribution with respect to $G$ and $\disc_{\mu}(G) = n^{-\Omega(1)}$, then
\begin{align*}
\QQ(f) & = O(n),\\
\QQ^{cc, *}(f \circ \twobitgadget) & = \Omega(n \log n).
\end{align*}
\end{theorem}

The proofs of Theorem~\ref{theo: transitive upp lower bound intro} and Theorem~\ref{theo: recipe for constructing BCW tight functions} each involve proving a query complexity upper bound and a communication complexity lower bound.
The proofs of the query complexity upper bounds are along similar lines and follow from Theorem~\ref{theo: query complexity doesn't increase on composition with hadamardization} and Corollary~\ref{cor: query complexity doesn't increase with parameters plugged in} (see Section~\ref{sec:query upper bound}). The proofs of the communication complexity lower bounds each involve a reduction from a problem whose communication complexity is easier to analyze (see Lemma~\ref{lm: reduction of communication problem from hadamardization} in Section~\ref{sec:communication lower bound}). Finally, we complete the proofs of Theorem~\ref{theo: transitive upp lower bound intro} and Theorem~\ref{theo: recipe for constructing BCW tight functions} in Section~\ref{sec:bcw tight}.

\subsection{Quantum query complexity upper bound}
\label{sec:query upper bound}

We start by stating the main theorem in this section. Recall from Definition~\ref{defi: hadamardization} that the function $h_G : \pmone^{j + k} \to \pmone$ is defined as $h_G(x, y) = G(s, t)$ if $x \in \pm H(s)$ and $y \in \pm H(t)$ for some $s \in \pmone^{\log j}$ and $t \in \pmone^{\log k}$, and $h_G(x, y) = \star$ otherwise.

Also recall from Definition~\ref{defi: composition with partial} that the function $r \wtcirc h_G : \pmone^{n(j+k)} \to \pmone$ is defined as $r \wtcirc h_G((X_1, Y_1), \dots, (X_n, Y_n)) = r(h_G((X_1, Y_1), \dots, (X_n, Y_n)))$ if $h_G((X_i, Y_i)) \in \pmone$ for all $i \in [n]$, and $-1$ otherwise.

\begin{theorem}
\label{theo: query complexity doesn't increase on composition with hadamardization}
Let $G : \pmone^{\log j} \times \pmone^{\log k} \to \pmone$ and $r : \pmone^{n} \to \pmone$. Then the quantum query complexity of the function $r \wtcirc h_G : \pmone^{n(j+k)} \to \pmone$ is given by
\[
    \QQ(r \wtcirc h_G) = O(n + \sqrt{n(j+k)}).
\]
\end{theorem}

As a corollary we obtain the following on instantiating $j = k = n$ and $r$ as a Boolean function with quantum query complexity $\Theta(n)$ in Theorem~\ref{theo: query complexity doesn't increase on composition with hadamardization}.

\begin{corollary}
\label{cor: query complexity doesn't increase with parameters plugged in}
Let $G : \pmone^{\log n} \times \pmone^{\log n} \to \pmone$ be a non-constant function and let $r : \pmone^{n} \to \pmone$ be a total function with $\QQ(r) = \Theta(n)$. Then the quantum query complexity of the total function $r \wtcirc h_G : \pmone^{2n^2} \to \pmone$ is 
\[
\QQ(r \wtcirc h_G) = \Theta(n).
\]
\end{corollary}

\begin{proof}
The upper bound $\QQ(r \wtcirc h_G) = O(n)$ follows by plugging in parameters in Theorem~\ref{theo: query complexity doesn't increase on composition with hadamardization}. 

For the lower bound, we show that $\QQ(r \wtcirc h_G) \geq \QQ(r)$.
Since $G$ is non-constant, there exist $x_1, y_1, x_2, y_2 \in \pmone^{\log n}$ such that $G(x_1, y_1) = -1$ and $G(x_2, y_2) = 1$. Let $X_1 = H(x_1), Y_1 = H(y_1)$, $X_2 = H(x_2)$ and $X_2 = H(y_2)$. Consider $r \wtcirc h_G$ only restricted to inputs where the inputs to each copy of $h_G$ are either $(X_1, Y_1)$ or $(X_2, Y_2)$.
Under this restriction, $r \wtcirc h_G : \pmone^{2n^2} \to \pmone$ is the same as $r : \pmone^{n} \to \pmone$. Thus $\QQ(r \wtcirc h_G) \geq \QQ(r) = \Omega(n)$.
\end{proof}

We now prove Theorem~\ref{theo: query complexity doesn't increase on composition with hadamardization}.

\begin{proof}[Proof of Theorem~\ref{theo: query complexity doesn't increase on composition with hadamardization}]
\textbf{Quantum query algorithm:}
View inputs to $r \wtcirc h_G$ as $(X_1, Y_1, \dots, X_n, Y_n)$, where $X_i \in \pmone^j$ for all $i \in [n]$ and $Y_i \in \pmone^k$ for all $i \in [n]$. We give a quantum algorithm and its analysis below.
\begin{enumerate}
    \item \label{item: step bernstein vazirani} Run $2n$ instances of the Bernstein-Vazirani algorithm: $1$ instance on each $X_i$ and $1$ instance on each $Y_i$, to obtain $2n$ strings $x_1, \dots, x_n, y_1, \dots, y_n$, where each $x_i$ is a $(\log j)$-bit string and each $y_i$ is a $(\log k)$-bit string.
    
    \item \label{item: step query empty set} For each $X_i$ and $Y_i$, query $(X_i)_{1^{\log j}}$ and $(Y_i)_{1^{\log k}}$ to obtain bits $b_i, c_i \in \pmone$ for all $i \in [n]$.
    
    \item \label{item: step Grover's search} In the next step, we run Grover's search~\cite{Grover96,BHMT02} to check equality of the following two $(nj + nk)$-bit strings: $(b_1 H(x_1), \dots, b_n H(x_n), c_1 H(y_1), \dots, c_n H(y_n))$ and $(X_1, \dots, X_n, Y_1, \dots, Y_n)$.
    
    \item \label{item: step output} If the step above outputs that the strings are equal, then output $r(G(x_1, y_1), \dots, G(x_n, y_n))$.  Else, output $-1$.
\end{enumerate}
    
\textbf{Analysis of the algorithm:}
    \begin{itemize}
        \item If the input is indeed of the form $(X_1, Y_1), \dots, (X_n, Y_n)$ where each $X_i \in \pm H(x_i)$ and $Y_i \in \pm H(y_i)$ for some $x_i \in \pmone^{\log j}$ and $y_i \in \pmone^{\log k}$, then Step~\ref{item: step bernstein vazirani} outputs the correct strings $x_{1}, \dots, x_{n}, y_1, \dots, y_n$ with probability $1$ by the properties of the Bernstein-Vazirani algorithm. Step~\ref{item: step query empty set} then implies that $X_i = b_i H(x_i)$ and $Y_i = c_i H(y_i)$ for all $i \in [n]$. Next, Step~\ref{item: step Grover's search} outputs that the strings are equal with probability 1 (since the strings whose equality is to be checked are equal). Hence the algorithm is correct with probability 1 in this case, since $(r \wtcirc h_G)(X_1, Y_1, \dots, X_n, Y_n) = r(G(x_1, y_1), \dots, G(x_n, y_n))$.

        \item If the input is such that there exists an index $i \in [n]$ for which $X_i \notin \pm H(x_i)$ for every $x_i \in \pmone^{\log j}$ or $Y_i \notin \pm H(y_i)$ for every $y_i \in \pmone^{\log k}$, then the two strings for which equality is to be checked in the Step~\ref{item: step Grover's search} are not equal.  Grover's search catches a discrepancy with probability at least $2/3$.  Hence, the algorithm outputs $-1$ (as does $r \wtcirc h_G$), and is correct with probability at least $2/3$ in this case.
\end{itemize}

\paragraph{Cost of the algorithm:} Step~\ref{item: step bernstein vazirani} accounts for $2n$ quantum queries.
Step~\ref{item: step query empty set} accounts for $2n$ quantum queries.
Step~\ref{item: step Grover's search} accounts for $O(\sqrt{n(j+k)})$ quantum queries.

Thus,
\[
    \QQ(r \wtcirc h_G) = O(n + \sqrt{n(j+k)}).
\]
\end{proof}

\subsection{Quantum communication complexity lower bound}
\label{sec:communication lower bound}
In this section we first show a communication lower bound (under some model) on $(r \wtcirc h_G) \circ \twobitgadget$ in terms of the communication complexity of $r \circ G$ (in the same model of communication) using a simple reduction. We state the lemma below (Lemma~\ref{lm: reduction of communication problem from hadamardization}) for the case where the models under consideration are the bounded-error and unbounded-error quantum models, since these are the models of interest to us.

\begin{lemma}
\label{lm: reduction of communication problem from hadamardization}
Let $r : \pmone^n \to \pmone$, $G : \pmone^{\log j} \times \pmone^{\log k} \to \pmone$, $\twobitgadget\in\{\AND_2,\XOR_2\}$ and $CC \in \{\QQ^{cc, *}, \UPP^{cc}\}$. Then,
\[
CC((r \wtcirc h_G) \circ \twobitgadget) \geq CC(r \circ G).
\]
\end{lemma}
\begin{proof}
We first consider the case $\twobitgadget = \AND_2$.
Consider a protocol $\Pi$ of cost $\ell$ that solves $(r \wtcirc h_G) \circ \twobitgadget$ in the $CC$-model. We exhibit below a protocol of cost $\ell$ that solves $r \circ G$ in the same model.

Suppose Alice is given input $x = (x_1, \dots, x_{n}) \in \pmone^{n \log j}$ and Bob is given input $y = (y_1, \dots, y_n) \in \pmone^{n \log k}$, where $x_i \in \pmone^{\log j}, y_i \in \pmone^{\log k}$ for each $i \in [n]$.

\paragraph{Preprocessing step:} 
Alice constructs the $(n(j+k))$-bit string 
\begin{equation}
\label{eq:alice input reduction}
X = ((H(x_1), (-1)^{k}), \dots, (H(x_n), (-1)^k)) \in \pmone^{n(j+k)},
\end{equation}
and Bob constructs the $(n(j+k))$-bit string \begin{equation}
\label{eq:bob input reduction}
Y = (((-1)^j, H(y_1)), \dots, ((-1)^j, H(y_n))) \in \pmone^{n(j+k)}.
\end{equation}

\paragraph{Protocol:} 
Alice and Bob run the protocol $\Pi$ with input $(X, Y)$ and output $\Pi(X, Y)$.

\paragraph{Cost:}
The preprocessing of the inputs to obtain $X$ from $x$, and $Y$ from $y$, takes no communication. Hence the total amount of communication is at most the cost of $\Pi$.

\paragraph{Correctness:}
For $X$ and $Y$ constructed in  Equation~\eqref{eq:alice input reduction} and Equation~\eqref{eq:bob input reduction}, respectively, we now argue that $(r \circ G)(x, y) = ((r \wtcirc h_G) \circ \AND_2) (X, Y)$, which would conclude the proof for $\twobitgadget = \AND_2$.
\begin{align*}
    ((r \wtcirc h_G) \circ \AND_2)(X, Y) & = (r \wtcirc h_G)((H(x_1), H(y_1)), \dots, (H(x_n), H(y_n)))\\
    & = r(G(x_1, y_1), \dots, G(x_n, y_n)) \tag*{by Definition~\ref{defi: hadamardization} and Definition~\ref{defi: composition with partial}}\\
    & = (r \circ G)(x, y).
\end{align*}
Thus,
\begin{align*}
CC((r \wtcirc h_G) \circ \AND_2) \geq CC(r \circ G).
\end{align*}
The argument for $\twobitgadget = \XOR_2$ follows along the same lines, with the strings $(-1)^j$  and $(-1)^k$ replaced by $1^j$  and $1^k$, respectively, in the preprocessing step. 
\end{proof}

\subsection{On the tightness of the BCW simulation}
\label{sec:bcw tight}
We prove Theorem~\ref{theo: recipe for constructing BCW tight functions} in Section~\ref{sec:bcw tight bounded error} and Theorem~\ref{theo: transitive upp lower bound intro} in Section~\ref{sec:bcw tight unbounded error}.

\subsubsection{Proof of Theorem~\ref{theo: recipe for constructing BCW tight functions}}
\label{sec:bcw tight bounded error}

Towards proving Theorem~\ref{theo: recipe for constructing BCW tight functions}, we first observe how to obtain bounded-error quantum communication complexity lower bounds using Theorem~\ref{theo:  discrepancy lb on quantum communication} and Lemma~\ref{lm: reduction of communication problem from hadamardization}.

\begin{lemma}
\label{lm:communication lower bound for hadamardized composed function in terms of adeg and disc}
Let $r : \pmone^n \to \pmone$ and $G : \pmone^{\log j} \times \pmone^{\log k} \to \pmone$ be functions such that $\bdisc(G) = o(1)$ and $\frac{8en}{\adeg(r)} \leq \left(\frac{1}{\bdisc(G)}\right)^{1-\beta}$ for some constant $\beta \in (0,1)$. Let $\twobitgadget\in\{\AND_2,\XOR_2\}$.
Then,
\[
\QQ^{cc, *} ((r \wtcirc h_G) \circ \twobitgadget) = \Omega\left(\adeg(r) \log\left(\frac{1}{\bdisc(G)}\right)\right).
\]
\end{lemma}

\begin{proof}
By Lemma~\ref{lm: reduction of communication problem from hadamardization} we have $\QQ^{cc, *} ((r \wtcirc h_G) \circ \twobitgadget) \geq \QQ^{cc, *}(r \circ G)$. By Theorem~\ref{theo:  discrepancy lb on quantum communication}, $\QQ^{cc, *}(r \circ G) = \Omega\left(\adeg(r) \log\left(\frac{1}{\bdisc(G)}\right)\right)$.
\end{proof}

We now prove Theorem~\ref{theo: recipe for constructing BCW tight functions}.

\begin{proof}[Proof of Theorem~\ref{theo: recipe for constructing BCW tight functions}]
Let $r : \pmone^n \to \pmone$, $G : \pmone^{\log n} \times \pmone^{\log n} \to \pmone$ and $f = r \wtcirc h_G : \pmone^{2n^2} \to \pmone$ be as in the statement of the theorem. We have $\QQ(r) \geq \adeg(r)/2 = \Omega(n)$, where the first inequality follows by Theorem~\ref{theo: bbc+01} and the second equality follows from the assumption that $\adeg(r) = \Omega(n)$. Moreover, $\QQ(r) \leq n$ since $r$ is a function on $n$ input variables. Hence $\QQ(r) = \Theta(n)$. Thus, Corollary~\ref{cor: query complexity doesn't increase with parameters plugged in} is applicable, and we have 
\[
\QQ(f) = \Theta(n).
\]
For the lower bound, $\adeg(r) = \Omega(n)$ by assumption. Thus 
\begin{align*}
    \frac{2en}{\adeg(r)} = O(1).
\end{align*}
Also, since by assumption $\frac{1}{\bdisc(G)} = n^{\Omega(1)} = \omega(1)$, we have
\begin{align*}
    \frac{2en}{\adeg(r)} \leq \left(\frac{1}{\bdisc(G)}\right)^{1-\beta}
\end{align*}
for every constant $\beta \in (0,1)$.
Lemma~\ref{lm:communication lower bound for hadamardized composed function in terms of adeg and disc} implies
\[
    \QQ^{cc, *}(f \circ \twobitgadget) = \Omega\left(\adeg(r) \log\left(\frac{1}{\bdisc(G)}\right) \right) = \Omega(n \log n). \qedhere
\]
\end{proof}

\subsubsection{Proof of Theorem~\ref{theo: transitive upp lower bound intro}}
\label{sec:bcw tight unbounded error}

The total function $f: \pmone^{2n^2} \to \pmone$ that we use to prove Theorem~\ref{theo: transitive upp lower bound intro} is $f = r \wtcirc h_G$, where $r = \PARITY_n$ and $G = \IP_{\log n}$. 
Note that Theorem~\ref{theo: recipe for constructing BCW tight functions} implies $\QQ(f) = O(n)$. 
For the quantum communication complexity lower bound in Theorem~\ref{theo: transitive upp lower bound intro}, we are able to show not only $\QQ^{cc, *}(f \circ \twobitgadget) = \Omega(n \log n)$, but $\UPP^{cc}(f \circ \twobitgadget) = \Omega(n \log n)$ for $\twobitgadget\in\{\AND_2,\XOR_2\}$. 
The following claim shows that $f$ is transitive.

\begin{claim}
\label{clm:func_is_transitive}
Let $n > 0$ be a power of 2. Let $r = \PARITY_n : \pmone^n \to \pmone$ and $G = \IP_{\log n} : \pmone^{\log n} \times \pmone^{\log n} \to \pmone$. The function $f = r \wtcirc h_G : \pmone^{2n^2} \to \pmone$ is transitive.
\end{claim}
\begin{proof}
We first show that $h_G : \pmone^{2n} \to \pmone$ is transitive. We next observe that $s \wtcirc t$ is transitive whenever $s$ is symmetric and $t$ is transitive ($s$ can be assumed to just be transitive rather than symmetric, as noted in Remark~\ref{rmk: transitive remark}). The theorem then follows since $\PARITY_n$ is symmetric.

Towards showing transitivity of $h_G$, let $\pi\in S_{2n}$, and $(\sigma_{\ell}, \sigma_\ell) \in S_{2n}$ for $\ell \in \pmone^{\log n}$ be defined as follows. (Here $\sigma_\ell \in S_n$; the first copy acts on the first $n$ coordinates, and the second copy acts on the next $n$ coordinates.)
\begin{itemize}
    \item 
    \[
    \pi(k) = \begin{cases}
    k + n & k \leq n\\
    k - n & k > n.
    \end{cases}
    \]
    That is, on every string $(x, y) \in \pmone^{2n}$, the permutation $\pi$ maps $(x, y)$ to $(y, x)$.
    
    \item For every $\ell \in \pmone^{\log n}$, the permutation $\sigma_\ell \in S_n$ is defined as
    \begin{equation}\label{eqn: sigma 1}
    \sigma_{\ell}(i) = i \oplus \ell,
    \end{equation}
    where $i \oplus \ell$ denotes the bitwise XOR of the strings $i$ and $\ell$. That is, for every input $(x, y) \in \pmone^{2n}$ and every $k \in \pmone^{\log n}$, the input bit $x_{k}$ is mapped to $x_{k \oplus \ell}$ and $y_k$ is mapped to $y_{k\oplus \ell}$.
\end{itemize}
For every $(x,y) \in \pmone^{2n}$ and $i,j \in \pmone^{\log n}$, the permutation $\sigma_{i \oplus j}(x,y)$ swaps $x_i$ and $x_j$, and also swaps $y_i$ and $y_j$. If for $i, j \in \pmone^{\log n}$, our task was to swap the $i$'th index of the first $n$ variables with the $j$'th index of the second $n$ variables, then the permutation $\sigma_{i\oplus j} \circ \pi$ does the job. That is, for every $(x,y) \in \pmone^{2n}$ and $i,j \in \pmone^{\log n}$, the permutation $\sigma_{i\oplus j} \circ \pi$ maps $x_i$ to $y_j$. Thus the set of permutations $\{\pi, \{\sigma_\ell : \ell \in \pmone^{\log n}\}\}$ acts transitively on $S_{2n}$.

Now we show that for all $x,y \in \pmone^{2n}$ and all $\ell \in \pmone^{\log n}$, we have $h_G(\sigma_{\ell}(x), \sigma_\ell(y)) = h_G(x,y)$. Fix $\ell \in \pmone^{\log n}$.
\begin{itemize}
    \item If $x \in \pm H(s)$ and $y \in \pm H(t)$ are Hadamard codewords, then $x_k = \langle k, s \rangle$ and $y_k = \langle k, t \rangle$ for all $k \in \pmone^{\log n}$, and $G(x,y) = \langle s, t \rangle$. Thus, for every $k \in \pmone^{\log n}$ we have $\sigma_{\ell}(x_k) = x_{k \oplus \ell} = \langle k \oplus \ell, s \rangle = \langle \ell, s \rangle \cdot \langle k , s \rangle$. Hence $\sigma_{\ell}(x) \in \pm H(s)$ (since $\langle \ell, s \rangle$ does not depend on $k$, and takes value either $1$ or $-1$). Similarly, $\sigma_{\ell}(y) \in \pm H(t)$. Thus  $h_G(\sigma_{\ell}(x,y)) = h_G(x,y)$.
    
    \item If $x$ ($y$, respectively) is not a Hadamard codeword, then a similar argument shows that for all $\ell \in [n]$, $\sigma_{\ell}(x)$ ($\sigma_{\ell}(y)$, respectively) is also not a Hadamard codeword.
\end{itemize}
Using the fact that $\langle s, t \rangle = \langle t, s \rangle$ for every $s, t \in \pmone^{\log n}$, one may verify that $h_G(\pi(x,y)) = h_G(x,y)$ for all $x,y \in \pmone^{2n}$.

Along with the observation that $\PARITY_n$ is a symmetric function, we have that $f = r \wtcirc h_G : \pmone^{2n^2} \to \pmone$ is transitive under the following permutations:
\begin{itemize}
    \item $S_n$ acting on the inputs of $\PARITY_n$, and
    \item The group generated by $\cbra{\pi} \cup \cbra{(\sigma_\ell, \sigma_\ell) : \ell \in [n]}$ acting independently on the inputs of each copy of $h_G$, where $\sigma_{\ell}$ is as in Equation~\eqref{eqn: sigma 1}. \qedhere
\end{itemize}
\end{proof}

We now prove Theorem~\ref{theo: transitive upp lower bound intro}.

\begin{proof}[Proof of Theorem~\ref{theo: transitive upp lower bound intro}]
Let $n > 0$ be a power of 2. Let $r = \PARITY_n : \pmone^n \to \pmone$ and $G = \IP_{\log n} : \pmone^{\log n} \times \pmone^{\log n} \to \pmone$. Let $f = r \wtcirc h_G : \pmone^{2n^2} \to \pmone$.
By Claim~\ref{clm:func_is_transitive}, $f$ is transitive.
By Corollary~\ref{cor: query complexity doesn't increase with parameters plugged in} we have
\[
\QQ(f) = \Theta(n).
\]
For the communication lower bound we have 
\begin{align*}
\UPP^{cc}(f \circ \twobitgadget) & = \UPP^{cc}((r \wtcirc h_G) \circ \twobitgadget) \\
& \geq \UPP^{cc}(\PARITY_n \circ \IP_{\log n}) \tag*{by Lemma~\ref{lm: reduction of communication problem from hadamardization}}\\
& = \UPP^{cc}(\IP_{n \log n}) \tag*{Observation~\ref{obs: parity composed inner product}}\\
& = \Omega(n \log n) \tag*{by Theorem~\ref{theo: forster} \qedhere}. 
\end{align*}
\end{proof}

\begin{remark}
\label{rmk: transitive remark}
The proof of transitivity of $f$ in Theorem~\ref{theo: transitive upp lower bound intro} can also be used to prove that if $r : \pmone^n$ is transitive and $G = \IP_{\log n} : \pmone^{\log n} \times \pmone^{\log n} \to \pmone$, then $r \wtcirc h_G$ is transitive as well. By instantiating $r$ to a transitive function with approximate degree $\Omega(n)$ (e.g., Majority), Theorem~\ref{theo: recipe for constructing BCW tight functions} implies that the BCW simulation is tight w.r.t.\ the bounded-error communication model for a wide class of transitive functions.
\end{remark}

\section{A separation between log-approximate-spectral norm and approximate degree for a transitive function}
\label{sec: Appendix log-approximate-spectral norm and approximate degree transitive function}

We now discuss the implications of our result for the Entropy Influence Conjecture, which is an interesting question in the field of analysis of Boolean functions, posed by Friedgut and Kalai~\cite{FK}. This conjecture is wide open for general functions. A much weaker version of this conjecture is called the Min-Entropy Influence Conjecture.

\begin{conjecture}[Min-Entropy Influence Conjecture] For any Boolean function $f:\pmone^n \to \pmone$ there exists a non-zero Fourier coefficient $\widehat{f}(S)$ such that $$\log\left(1/|\widehat{f}(S)|\right) = O(I(f)),$$
where $I(f)$ denotes the influence of $f$ ($I(f) = \sum_{S\subseteq [n]}|S|\widehat{f}(S)^2$).
\end{conjecture}

While this conjecture is also wide open, some attempts have been made to prove various implications of this conjecture. One interesting implication of the Min-Entropy Influence Conjecture that is still open is whether the min-entropy of the Fourier spectrum (that is, $\log\left(1/\max_{S \subseteq [n]}|\widehat{f}(S)|\right)$) is $O(\QQ(f))$. In \cite{ACK+18} using a primal-dual technique it was shown that the min-entropy of the Fourier spectrum is at most a constant times $\log(\|\hat{f}\|_{1,\epsilon})$, where the constant depends on $\epsilon$.  Thus if it were the case that $\log(\|\hat{f}\|_{1,\epsilon}) = O(\QQ(f))$, we would have upper bounded the min-entropy of the Fourier spectrum by $O(\QQ(f))$. This was stated in the conference version of~\cite{ACK+18} as a possible approach towards proving the Min-Entropy Influence Conjecture, and was left as an open problem. 

\begin{question}[{\cite[Section 6]{ACK+18}}]\label{qn: ack+}
Is it true that for all Boolean functions $f : \pmone^n \to \pmone$,
\[
\log(\|\wh{f}\|_{1,\epsilon}) = O(\adeg(f))?
\]
\end{question}

We show here that their question has a positive answer for the class of symmetric Boolean functions (Section~\ref{sec:sym}), but a negative answer for general Boolean functions and in fact already for  the special case of transitive Boolean functions (Section~\ref{section: Negative answer}).

\subsection{Upper bound on the approximate spectral norm of symmetric functions}\label{sec:sym}

Recall from Definition~\ref{def: Symmetric Multilinear Polynomial} that a multilinear polynomial $\phi : \R^n \to \R$ is said to be symmetric if $\phi(x_1, \dots, x_n) = \phi(x_{\sigma(1)}, \dots, x_{\sigma(n)})$ for all $(x_1, \dots, x_n) \in \R^n$ and $\sigma \in S_n$. Sherstov~\cite{She18} showed the following upper bound on the spectral norm of symmetric multilinear polynomials.
\begin{claim}[{\cite[Lemma 2.9]{She18}}]\label{claim:1norm_symmetric_multilinear}
Let $\phi : \R^n \to \R$ be a symmetric multilinear polynomial. Then 
\[
    \|\phi\|_1 \leq 8^{\deg(\phi)} \max_{x \in \pmone^n} |\phi(x)|.
\]
\end{claim}

\begin{lemma}
Let $f : \pmone^n \to \pmone$ be a symmetric Boolean function. Then 
\[
    \log(\|f\|_{1,1/3}) = O(\adeg(f)).
\]
\end{lemma}

\begin{proof}
Let $p' : \pmone^n \to \pmone$ be a polynomial of degree $O(\adeg(f))$ that $1/3$-approximates the symmetric function $f$. Let $S_n$ be the set of all permutations of $[n]$, then the polynomial 
\begin{align*}
    p(x_1, \dots, x_n) = \frac{1}{n!} \sum_{\sigma \in S_n} p'(x_{\sigma(1)}, \dots , x_{\sigma(n)})
\end{align*}
also $1/3$-approximates $f$. Observe that $\deg(p) 
 \leq \deg(p')$ and $p$ is a symmetric multilinear polynomial. Thus 
 \begin{align*}
     \log\left( \|p\|_1 \right) = O(\deg(p)) = O(\adeg(f)),
 \end{align*}
where the first equality follows from Claim~\ref{claim:1norm_symmetric_multilinear} and the second follows from the choice of~$p'$.
\end{proof}

\subsection{Negative answer to Question~\ref{qn: ack+}}
\label{section: Negative answer}

In this section we show that  Question~\ref{qn: ack+} has a negative answer even for the special class of transitive Boolean functions. 

\begin{claim}
\label{clm: ack+ if false for transitive functions}
There exists a transitive function $f: \pmone^{2n^2} \to \pmone$ such that $\log\left(\|\wh{f}\|_{1, 1/3} \right) = \Omega(\adeg(f) \log n)$.
\end{claim}

We first state some required preliminaries.

\begin{definition}[Monomial projection]
\label{defi: monomial projection}
We call a function $g:\pmone^m \to \pmone$ a monomial projection of a function $f:\pmone^n \to \pmone$ if $g$ can be expressed as $g(x_1,\dots, x_m) = f(M_1,\dots, M_n)$, where each $M_i$ is a monomial in the variables $x_1,\dots, x_m$.
\end{definition}

It is known that the approximate spectral norm of a function cannot increase upon monomial projections (see, for example, \cite[Observation 25]{CM17}).
\begin{observation}[{\cite[Observation 25]{CM17}}]
\label{obs: monproj weight reduction}
For $f : \pmone^{n} \to \pmone$ and $g : \pmone^{m} \to \pmone$ such that $g$ is a monomial projection of $f$,
\[
\|\wh{g}\|_{1, 1/3} \leq \|\wh{f}\|_{1, 1/3}.
\]
\end{observation}

For a proof of the following fact, see~\cite[Section 6.3]{odonnell:analysis}.

\begin{fact}[Fourier coefficients of $\IP_{n}$]
\label{fact: ip fourier coefficients}
Let $\IP_n: \pmone^{2n} \to \pmone$ be as in Definition~\ref{defi: inner product function}. Then for all $S \subseteq[2n]$ we have,
\begin{align*}
    \abs{\wh{\IP_n}(S)} = \frac{1}{2^n}.
\end{align*}
\end{fact}

\begin{proof}[Proof of Claim~\ref{clm: ack+ if false for transitive functions}]
Let $n > 0$ be a power of 2. Let $r = \PARITY_n : \pmone^n \to \pmone$ and $G = \IP_{\log n} : \pmone^{\log n} \times \pmone^{\log n} \to \pmone$. From Claim~\ref{clm:func_is_transitive}, the function $f = r \wtcirc h_G : \pmone^{2n^2} \to \pmone$ is transitive. By Corollary~\ref{cor: query complexity doesn't increase with parameters plugged in} we have
\begin{align}
    \QQ(f) = \Theta(n). \label{eq: approx-norm quantum ub.}
\end{align}
By Theorem~\ref{theo: bbc+01} we have $\QQ(f) \geq \adeg(f)/2$ and together with Equation~\eqref{eq: approx-norm quantum ub.}, this implies $\adeg(f) = O(n)$. Thus to complete the proof of the claim, it suffices to show $\log\left(\|\wh{f}\|_{1, 1/3} \right) = \Omega(n \log n)$.

We first note that $\IP_{n \log n}$ is a monomial projection of $f$. Consider the function $f$ acting on the input variables $x^{(1)}, \dots, x^{(n)}, y^{(1)}, \dots, y^{(n)}$, where $x^{(i)}, y^{(i)} \in \pmone^{n}$ for all $i \in [n]$.

For $i \in [n]$, set $x^{(i)}_{1^{\log n}} = y^{(i)}_{1^{\log n}} = 1$.
For $i \in [n]$ and string $s \in \pmone^{\log n} \setminus \cbra{1^{\log n}}$, set $x^{(i)}_s = \prod_{j: s_j = -1}x^{(i)}_j $ and $y^{(i)}_s = \prod_{j: s_j = -1}y^{(i)}_j$ .
That is, in each block of inputs $x^{(i)}$ and $y^{(i)}$, the coordinate corresponding to $1^{\log n}$ equals 1, the coordinates corresponding to $j \in \pmone^{\log n}$ with $|j| = 1$ are free variables, and all other variables are replaced by monomials in these variables.
Under this monomial projection there are $2n \log n$ free variables, namely $\left\{x^{(i)}_{s}, y^{(i)}_{s} : i \in [n], s \in \pmone^{\log n}, |s| = 1\right\}$. Also note that under this projection and every setting of the free variables, the blocks $x^{(i)}$ and $y^{(i)}$, for $i \in [n]$, are always Hadamard codewords. Let $f'$ be the monomial projection (see Definition~\ref{defi: monomial projection}) of $f$ under the projection defined in this paragraph. For the purpose of the next equality we abbreviate strings $s \in \pmone^{\log n}$ of Hamming weight 1 by the set $\cbra{i}$, where $i \in [\log n]$ is such that $s_i = -1$. On the free variables, the projected function $f'$ equals
\[
\PARITY_n \left(\langle (x^{(1)}_{\{1\}}, \dots, x^{(1)}_{\{\log n\}}), (y^{(1)}_{\{1\}}, \dots, y^{(1)}_{\{\log n\}})\rangle, \dots, \langle (x^{(n)}_{\{1\}}, \dots, x^{(n)}_{\{\log n\}}), (y^{(n)}_{\{1\}}, \dots, y^{(n)}_{\{\log n\}})\rangle\right).
\]
Thus,
\[
f' = \IP_{n\log n}.
\]

It follows from earlier works~\cite{Bruck90, BS92} that every polynomial that approximates $\IP_{n\log n}$ to error $1/3$ must have spectral norm $2^{\Omega(n \log n)}$. We include a short proof below for completeness.

Let $P$ be a polynomial that approximates $\IP_{n\log n}$ to error $1/3$. Since $P(x, y)\IP_{n\log n}(x, y) \geq 2/3$ for all $(x, y) \in \pmone^{2n\log n}$, we have
\begin{align*}
    2/3 & \leq \mathbb{E}_{x, y \in \pmone^{2n\log n}} [P(x, y)\IP_{n\log n}(x, y)]\\
    & = \sum_{S \subseteq [2n\log n]}\widehat{P}(S)\widehat{\IP_{n\log n}}(S) \tag*{by Plancherel's theorem (Fact~\ref{fact: plancherel's thm})}\\
    & \leq \sum_{S \subseteq [2n\log n]}\frac{|\widehat{P}(S)|}{2^{n \log n}} \tag*{by Fact~\ref{fact: ip fourier coefficients}}\\
    & = \frac{\|\wh{P}\|_1}{2^{n \log n}}\\
    & \implies \log(\|\wh{P}\|_1) = \Omega(n \log n)\\
    & \implies \log (\|\wh{\IP_{n \log n}}\|_{1, 1/3}) = \Omega(n \log n).
\end{align*}
This yields the desired contradiction by Observation~\ref{obs: monproj weight reduction}.
\end{proof}

\appendix

\section{Quantum communication lower bound via the  generalized discrepancy method}
\label{sec: appendix communication complexity lower bound via generalized discrepancy method}

In this appendix we prove Theorem~\ref{theo:  discrepancy lb on quantum communication}, which gives a lower bound on the quantum communication complexity of a composed function in terms of the approximate degree of the outer function and the discrepancy of the inner function. This result is implicit in~\cite[Theorem 7]{LZ10}. Their result is stated in the more general setting of matrix norms, and the $\log$ factor we require on the right-hand side of Theorem~\ref{theo:  discrepancy lb on quantum communication} is not included in their statement. Theorem~\ref{theo:  discrepancy lb on quantum communication} follows implicitly from the proof of~\cite[Theorem 7]{LZ10} along with the fact that the $\gamma_2^*$-norm characterizes discrepancy~\cite{LS09b}.

For completeness and clarity, we prove Theorem~\ref{theo:  discrepancy lb on quantum communication} below from first principles.

\begin{definition}
\label{defi: correlation definition}
For functions $f,g : \pmone^n \to \R$ and a probability distribution $\mu: \pmone^n \to \R$ the correlation between $f$ and $g$ with respect to $\mu$ is defined to be
\begin{align*}
    \corr{f}{g}{\mu} = \sum_{x \in \pmone^n} f(x)g(x) \mu(x).
\end{align*}
\end{definition}

For a Boolean function $f : \pmone^n \to \pmone$, its approximate degree can be captured by a certain linear program. Writing out the dual of this program and analyzing its optimum yields the following theorem (see, for example,~\cite[Theorem 1]{BT13}).
\begin{theorem}[Dual witness for $\epsilon$-approximate degree]
\label{theo: dual witness}
For every $\epsilon \geq 0$ and $f : \pmone^n \to \pmone$, $\adeg_{\epsilon}(f) \geq d$ if and only if there exists a polynomial $\psi:\pmone^n \to \R$ such that
\begin{enumerate}
    \item $\sum_{x \in \pmone^n} f(x) \psi(x) > \epsilon$,
    
    \item $\sum_{x \in \pmone^n} |\psi(x)| = 1$ and 
    
    \item $\wh{\psi}(S) = 0$ for all $|S| < d$.
\end{enumerate}
\end{theorem}

We require the following fact, which follows immediately from the fact that $\wh{f}(S)$ is a uniform average of different signed values of $f(x)$.
\begin{fact}[Folklore]
\label{fact: dual has small Fourier coefficients}
For every function $f: \pmone^n \to \R$,
\begin{align*}
    2^n \max_{S \subseteq[n]} \abs{\wh{f}(S)} \leq \sum_{x \in \pmone^n} \abs{f(x)}.
\end{align*}
\end{fact}

The following theorem shows that if a two-party function $F$ correlates well with a two-party function $G$ under some distribution $\lambda$, and the discrepancy of $G$ under $\lambda$ is small, then the bounded-error quantum communication complexity of $F$ must be large.
This is referred to as the generalized discrepancy method, and was first proposed by Klauck~\cite{Klauck07}. 
The following version can be found in~\cite[page~173]{Cha09}, for example, stated as a lower bound on randomized communication complexity. However, the generalized discrepancy method is also known to give lower bounds on bounded-error quantum communication complexity, even in the model where the parties can share an arbitrary prior entangled state for free.

\begin{theorem}[Generalized Discrepancy Bound]
\label{theo: generalized discrepancy bound}
Consider functions $E, F : \pmone^m \times \pmone^n \to \pmone$. If there exists a distribution $\lambda : \pmone^{m} \times \pmone^n \to \R$ such that $\corr{E}{F}{\lambda} \geq \delta$, then \[
\QQ^{cc, *}_{\frac12 - \epsilon}(E)=\Omega\left(  \log\left(\frac{\delta + 2\epsilon - 1}{\disc_{\lambda}(F)} \right)\right).
\]
\end{theorem}

\begin{definition}
For a distribution $\mu : \pmone^m \times \pmone^n \to \R$ and integer $k>0$, define the distribution $\mu^{\otimes k} : \pmone^{mk} \times \pmone^{nk} \to \R$ by
    \[
    \mu^{\otimes k}((X_1, Y_1) \dots, (X_k, Y_k)) = \prod_{i \in [k]}\mu(X_i, Y_i),
    \]
    where $X_i \in \pmone^{m}$ and $Y_i \in \pmone^n$ for all $i \in [k]$.
\end{definition}

We require the following XOR lemma for discrepancy due to Lee, Shraibman and {\v{S}}palek~\cite{LSS08}. \begin{theorem}[{\cite[Theorem 19]{LSS08}}]
\label{theo: lss08}
Let $P : \pmone^m \times \pmone^n \to \pmone$ be a two-party function and $\mu : \pmone^m \times \pmone^n \to \R$ be a distribution. For every integer $k > 0$,
\[
\disc_{\mu^{\otimes k}}(\PARITY_k \circ P) \leq (8 \disc_{\mu}(P))^{k}.
\]
\end{theorem}

We recall Theorem~\ref{theo:  discrepancy lb on quantum communication} below.
\begin{theorem}[Restatement of Theorem~\ref{theo:  discrepancy lb on quantum communication}]
\label{theo: appendix discrepancy lb on quantum communication}
Let $r : \pmone^n \to \pmone$ and $G : \pmone^j \times \pmone^k \to \pmone$ be functions.
Let $\mu: \pmone^j \times \pmone^k \to \R$ be a balanced distribution with respect to $G$ and $\disc_{\mu}(G) = o(1)$. If $\frac{8en}{\adeg(r)} \leq \left(\frac{1}{\disc_{\mu}(G)}\right)^{1-\beta}$ for some constant $\beta \in (0,1)$, then
\begin{align*}
    \QQ^{cc, *}(r \circ G) =  \Omega\left(\adeg(r) \log\left(\frac{1}{\disc_{\mu}(G)}\right)\right).
\end{align*}
In particular,
\begin{align*}
    \QQ^{cc, *}(r \circ G) =  \Omega\left(\adeg(r) \log\left(\frac{1}{\bdisc(G)}\right)\right).
\end{align*}
\end{theorem}

The following proof method is adopted from Chattopadhyay's thesis~\cite{Cha09}.

\begin{proof}
For simplicity we assume $j = k = m$. The general proof follows along similar lines. Let $\adeg(r) = d$, and let $\psi: \pmone^n \to \R$ be a dual witness for this as given by Theorem~\ref{theo: dual witness}. Let $\nu : \pmone^n \to \R$ be defined as $\nu(x)  = |\psi(x)|$ for all $x \in \pmone^n$ and also define $h: \pmone^n \to \pmone$ as $h(x) = \sign(\psi(x))$ for all $x \in \pmone^n$. First note that $\nu$ is a distribution since $\sum_{x \in \pmone^n} \nu(x) = \sum_{x \in \pmone^n} |\psi(x)| = 1$ by Theorem~\ref{theo: dual witness}. From Theorem~\ref{theo: dual witness} we also have
\begin{align}
    \corr{r}{h}{\nu} & > 1/3 \label{eq: dual property 2} \\
    \wh{h\nu}(S) & = 0 \quad \text{for all~} |S| < d, 
    \label{eq: dual property 3}
\end{align}
where $h\nu(x) := h(x)\nu(x) = \psi(x)$ for all $x \in \pmone^n$.
We will construct a probability distribution $\lambda : \pmone^{mn} \times \pmone^{mn} \to \R$ using $\nu : \pmone^n \to \R$ and $\mu: \pmone^m \times \pmone^m \to \R$ such that $r \circ G$ and $h \circ G$ have large correlation under $\lambda$. We will also show that $\disc_{\lambda}(h \circ G)$ is small. The proof of the theorem would then follow from Theorem~\ref{theo: generalized discrepancy bound}.

Let $X,Y \in \pmone^{mn}$ be such that $X = (X_1, \dots, X_n)$, $Y = (Y_1, \dots, Y_n)$ and $X_i, Y_i \in \pmone^m$ for all $i \in [n]$. Also, let $G(X,Y) = (G(X_1, Y_1), \dots, G(X_n, Y_n)) \in \pmone^n$.

Define $\lambda: \pmone^{mn} \times \pmone^{mn} \to \R$ as follows.
\begin{align}
\label{eq: define lambda}
    \lambda(X,Y)
    &= 2^n \cdot \nu(G(X, Y)) \cdot \prod_{i \in [n]} \mu(X_i,Y_i). 
\end{align}

Observe that for any $z \in \pmone^n$,
\begin{align}
    \sum_{X,Y : G(X, Y) = z} \lambda(X,Y)
    &= \sum_{X,Y: G(X,Y) = z} 2^n \cdot \nu(G(X, Y)) \cdot \prod_{i \in [n]} \mu(X_i,Y_i) \tag*{by Equation~\eqref{eq: define lambda}} \nonumber\\
    &= 2^n \nu(z) \prod_{i \in [n]}\left(\sum_{X_i, Y_i : G(X_i, Y_i) = z_i} \mu(X_i, Y_i)\right)\nonumber \\
    &= \nu(z), \label{eq: lambdasum}
\end{align}
where the last equality follows since $\mu$ is balanced w.r.t.~$G$ by assumption.
Thus
\[
    \sum_{X,Y} \lambda(X,Y) = \sum_{z \in \pmone^n} \nu(z) = 1.
\]
We next observe that $\corr{r \circ G}{h \circ G}{\lambda}$ is large.
\begin{align}
    \corr{r \circ G}{h \circ G}{\lambda}
    &= \sum_{X,Y} {r \circ G}(X,Y) \cdot {h \circ G}(X,Y) \cdot \lambda(X,Y) \nonumber \\
    &= \sum_{z \in \pmone^n} \sum_{X,Y: G(X,Y) = z} r(z)h(z) \cdot \lambda(X,Y) \nonumber \\
    &= \sum_{z \in \pmone^n} r(z)h(z) \nu(z) \tag*{by Equation~\eqref{eq: lambdasum}} \nonumber \\
    &= \corr{r}{h}{\nu} > 1/3. \label{eq: correlation r and h under nu}
\end{align}
where the last equality follows from  Equation~\eqref{eq: dual property 2}.

We now upper bound the discrepancy of $h \circ G$ with respect to a rectangle $R$, under the distribution $\lambda$.
Let $R \subseteq \pmone^{mn} \times \pmone^{mn}$ be any rectangle of the form $R(X, Y) = A(X)B(Y)$ for $A : \pmone^{mn} \to \zone$ and $B : \pmone^{mn} \to \zone$.
\begin{align*}
    \discrepancy{h \circ G}{R}{\lambda}
    &= \abs{\sum_{X,Y} {h \circ G}(X,Y) \cdot R(X,Y) \cdot \lambda(X,Y)} \\
    &= \abs{\sum_{z \in \pmone^n} \sum_{X,Y: G(X,Y) = z} {h \circ G}(X,Y) \cdot R(X,Y) \cdot \lambda(X,Y)}\\
    &= 2^n \abs{\sum_{z \in \pmone^n}  h(z) \nu(z) \sum_{X,Y: G(X,Y) = z} R(X,Y) \prod_{i \in [n]} \mu(X_i,Y_i)} \tag*{by Equation~\eqref{eq: define lambda}}\\
    &= 2^n \abs{\sum_{z \in \pmone^n}  \left(\sum_{S \subseteq [n]} \wh{h\nu}(S)\chi_S(z)\right) \sum_{X,Y: G(X,Y) = z} R(X,Y) \prod_{i \in [n]} \mu(X_i,Y_i)}\\
    &= 2^n\abs{\sum_{z \in \pmone^n} \sum_{S \subseteq [n]} \wh{h\nu}(S) \sum_{X,Y: G(X,Y) = z} R(X,Y) \cdot \chi_S(z) \cdot \prod_{i \in [n]} \mu(X_i,Y_i)} \\
    &= 2^n\abs{\sum_{z \in \pmone^n} \sum_{S \subseteq [n]} \wh{h\nu}(S) \sum_{X,Y: G(X,Y) = z} R(X,Y) \cdot \prod_{i \in S}G(X_i,Y_i) \cdot \prod_{i \in [n]} \mu(X_i,Y_i)} \\
    & \leq 2^n \sum_{S \subseteq [n], |S| \geq d}\abs{\wh{h\nu}(S)} \cdot
    \abs{
    \sum_{X,Y} R(X,Y) \cdot \prod_{i \in S}G(X_i,Y_i)\mu(X_i, Y_i) \cdot \prod_{j \notin S}\mu(X_j,Y_j)} \tag*{since $\wh{h\nu}(S) = 0$ for all $|S| < d$ by Equation~\eqref{eq: dual property 3}}\\
    & \leq 2^n \sum_{S \subseteq [n], |S| \geq d}\abs{\wh{h\nu}(S)} \cdot \abs{
    \sum_{X_j,Y_j : j \notin S} \prod_{j \notin S}\mu(X_j,Y_j) \left(\sum_{X_i,Y_i : i \in S} A(X) B(Y) \cdot \prod_{i \in S}G\mu(X_i, Y_i)\right)}.
    \end{align*}
    
    For any $X = (X_1, \dots, X_n) \in \pmone^{mn}$ (respectively, $Y = (Y_1, \dots, Y_n) \in \pmone^{mn}$) and set $S \subseteq [n]$ define the string $X_S \in \pmone^{m|S|}$ (respectively, $Y_S \in \pmone^{m|S|}$) by $X_S = (\dots, X_i, \dots)$ (respectively, $Y_S = (\dots, Y_i, \dots)$), where $i$ ranges over all elements of $S$. For any $S$ and fixed $\{X_j : j \notin S\}$ (respectively, $\{Y_j : j \notin S\}$) note that $A(X) = A'(X_S)$ (respectively, $B(X) = B'(X_S)$) for some $A': \pmone^{m|S|} \to \zone$ (respectively, $B': \pmone^{m|S|} \to \zone$) which is a function of the fixed values $\{X_j : j \notin S\}$ (respectively, $\{Y_j : j \notin S\}$). Let $R'(X_S, Y_S) = A'(X_S) B'(Y_S)$.

    Continuing from the above we obtain,
    \begin{align*}
    \discrepancy{h \circ G}{R}{\lambda}    
    &\leq 2^n \sum_{S \subseteq [n], |S| \geq d}\abs{\wh{h\nu}(S)} \abs{
    \sum_{X_j,Y_j : j \notin S} \prod_{j \notin S}\mu(X_j,Y_j) \left(\sum_{X_i,Y_i : i \in S} R'(X_S, Y_S) \prod_{i \in S}G\mu(X_i, Y_i)\right)}\\
    & \leq 2^n \sum_{S \subseteq [n], |S| \geq d}\abs{\wh{h\nu}(S)} \abs{ \sum_{X_j,Y_j : j \notin S} \prod_{j \notin S}\mu_j(X_j,Y_j) \left(\disc_{\mu^{\otimes |S|}}(\XOR_{|S|} \circ G)\right)} \\
    &= 2^n \sum_{S \subseteq [n], |S| \geq d}\abs{\wh{h\nu}(S)} \cdot
    \disc_{\mu^{\otimes |S|}}(\XOR_{|S|} \circ G)\\
    &\leq \sum_{S \subseteq [n], |S| \geq d} \left(8 \disc_{\mu}(G)\right)^{|S|}
    \tag*{from Fact~\ref{fact: dual has small Fourier coefficients} and Theorem~\ref{theo: lss08}}
    \\
    &= \sum_{k = d}^n \binom{n}{k} \left(8 \disc_{\mu}(G)\right)^{k} \leq \sum_{k = d}^n \left( \frac{8en}{k} \disc_{\mu}(G)\right)^{k}
    \\
    &\leq \sum_{k = d}^n \left(\disc_{\mu}(G)^{\beta}\right)^{k}
    \tag*{since $\frac{8en}{d} \disc_{\mu}(G) \leq \left(\disc_{\mu}(G)\right)^{\beta}$ by assumption}
    \\
     &\leq \frac{\disc_{\mu}(G)^{d \cdot \beta}}{1 - \disc_{\mu}(G)^{\beta}}
\end{align*}

Since $\disc_{\mu}(G) = o(1)$ by assumption ($\disc_{\mu}(G) = 1 - \Omega(1)$ suffices, but we assume $\disc_{\mu}(G) = o(1)$ in the statement for readability) and $\beta \in (0, 1)$ is a constant, we have $1 - \disc_\mu(G)^\beta = \Omega(1)$, and hence
\begin{align}
\label{eq: discrepancy of hoG wrt lambda}
\disc_{\lambda}(h \circ G) = O(\disc_{\mu}(G)^{d \beta}).
\end{align}

Hence,
\begin{align*}
    \QQ^{cc}_{1/2 - 2/5}(r \circ G) 
    &\geq  \log\left(\frac{1/3 + 4/5 - 1}{\disc_{\lambda}(h \circ G)}\right) \tag*{by Theorem~\ref{theo: generalized discrepancy bound}}\\
    &= \Omega\left( \log\left(\frac{1}{\disc_{\mu}(G)^{d\beta}}\right) \right) \tag*{by Equation~\eqref{eq: discrepancy of hoG wrt lambda}}\\
    &= \Omega\left(\adeg(r) \log\left( \frac{1}{\disc_{\mu}(g)}\right)\right) \tag*{since $\beta = \Omega(1)$ by assumption, and $d = \adeg(r)$}.
\end{align*}

The theorem follows since $\QQ^{cc}(F) = \Theta(\QQ^{cc}_{\epsilon}(F))$ for all constants $\epsilon \in (0, 1/2)$ and all Boolean functions $F$ (with the constant in the $O(\cdot)$ depending on~$\epsilon$).
\end{proof}

\paragraph{Acknowledgments} We thank the anonymous ToC reviewers for useful comments, and for spotting an issue in an earlier version of one of the proofs in Section~\ref{sec:noisyamplamplif}.

\bibliographystyle{alpha}
\bibliography{references}

\end{document}